%% file: lhwr_theory_TAS.tex
\documentclass[12pt]{article}

\usepackage[T1]{fontenc}
\usepackage[utf8]{inputenc}
\usepackage{lmodern}
\usepackage[letterpaper,margin=1in]{geometry}
\usepackage{amsmath,amssymb,amsthm,bm}
\usepackage{booktabs}
\usepackage{graphicx}
\usepackage{float}
\usepackage{placeins}
\usepackage[ruled,vlined,linesnumbered]{algorithm2e}
\usepackage{xspace}
\usepackage{url}
\usepackage[numbers,sort&compress]{natbib}
\usepackage[colorlinks=true,linkcolor=blue,citecolor=blue,urlcolor=blue]{hyperref}

\newtheorem{assumption}{Assumption}
\newtheorem{theorem}{Theorem}
\newtheorem{proposition}{Proposition}
\newtheorem{corollary}{Corollary}

\newtheorem{lemma}{Lemma}[section]

\newcommand{\R}{\mathbb{R}}
\newcommand{\Htwo}{\mathbb{H}^{2}}
\newcommand{\innerL}[2]{\left\langle #1,#2\right\rangle_L}
\newcommand{\dH}{d_{\mathbb H}}
\newcommand{\acosh}{\operatorname{arcosh}}
\newcommand{\argmin}{\operatorname*{arg\,min}}
\newcommand{\tr}{\operatorname{tr}}
\newcommand{\diag}{\operatorname{diag}}
\newcommand{\E}{\mathbb{E}}
\newcommand{\Var}{\operatorname{Var}}
\newcommand{\Cov}{\operatorname{Cov}}
\newcommand{\RSS}{\ensuremath{\mathrm{RSS}}\xspace}
\newcommand{\AIC}{\ensuremath{\mathrm{AIC}}\xspace}
\newcommand{\AICc}{\ensuremath{\mathrm{AICc}}\xspace}
\newcommand{\CV}{\ensuremath{\mathrm{CV}}\xspace}
\newcommand{\RMSE}{\ensuremath{\mathrm{RMSE}}\xspace}
\newcommand{\MAE}{\ensuremath{\mathrm{MAE}}\xspace}

\newcommand{\by}{\bm y}
\newcommand{\bX}{\bm X}
\newcommand{\bZ}{\bm Z}
\newcommand{\bbeta}{\bm\beta}
\newcommand{\bx}{\bm x}
\newcommand{\bW}{\bm W}
\newcommand{\bA}{\bm A}
\newcommand{\bG}{\bm G}
\newcommand{\be}{\bm e}
\newcommand{\bI}{\bm I}
\newcommand{\bS}{\bm S}
\newcommand{\bepsilon}{\bm\varepsilon}
\newcommand{\bs}{\bm s}

\hypersetup{
  pdftitle={Lorentz Hyperbolic Weighted Regression: Theory for Fixed and Estimated Representations},
  pdfauthor={Bahadır Yüzbaşı; Zühal Küçükarslan Yüzbaşı},
  pdfkeywords={hyperbolic representations; local regression; Lorentz model; varying coefficients; diagnostics; economic similarity}
}

\title{Lorentz Hyperbolic Weighted Regression:\\
Theory for Fixed and Estimated Representations}
\author{Bahadır Yüzbaşı\textsuperscript{1,*} \and
Zühal Küçükarslan Yüzbaşı\textsuperscript{2}\\[0.5em]
\small \textsuperscript{1}Department of Econometrics, İnönü University,
Malatya 44280, Türkiye\\
\small \textsuperscript{2}Department of Mathematics, Fırat University,
Elazığ, Türkiye\\
\small \textsuperscript{*}Corresponding author:
\href{mailto:bahadir.yuzbasi@inonu.edu.tr}{bahadir.yuzbasi@inonu.edu.tr}}
\date{}

\begin{document}

\raggedbottom
\setlength{\emergencystretch}{2em}
\maketitle

\begin{abstract}
Many applications provide each observation with a meaningful representation
in addition to an ordinary response and covariates. When that representation
is hierarchical, hub--periphery structured, or network derived, Euclidean or
geographic proximity may define the wrong peer groups. We present Lorentz
hyperbolic weighted regression (LHWR) as a practical local regression method
for this setting. Responses and predictors remain real valued; only locality
is defined by distances between observations represented on the Lorentz model
of hyperbolic space. We describe coordinate construction, adaptive bandwidth
selection, prediction, local coefficient summaries, collinearity checks, and
residual autocorrelation diagnostics. Theoretical results explain
consistency, bias--variance tradeoffs, curvature effects, and the extra
uncertainty caused by estimated representations. Simulations show that the
Lorentz geometry is most useful for sharply localized coefficient surfaces,
whereas tangent-plane approximations can be competitive for smooth surfaces.
In a 141-country World Development Indicators illustration, economic
similarity based on income and trade openness defines local peer groups.
Representation-based locality predicts GDP growth better than global least
squares and geographic local regression, although differences among Lorentz,
Poincar\'e, and tangent metrics are modest. The main practical lesson is that
the representation should be chosen scientifically and the distance geometry
should be checked rather than assumed.
\end{abstract}


\section{Introduction}
\label{sec:introduction}

Local regression is usually introduced through physical location. Nearby
observations receive greater weight, and the fitted association is allowed to
vary over space. This principle is useful far beyond geography. A hospital may
want to borrow information from clinically similar patients, a university
from institutions with comparable research profiles, or a country from
economies with similar income and trade integration. In each example, the
scientifically relevant neighbourhood is defined by a representation rather
than longitude and latitude.

Hyperbolic representations are particularly attractive when the data exhibit
hierarchy, hubs, or rapidly expanding neighbourhoods
\citep{krioukov2010,papadopoulos2012,nickel2017,nickel2018}. Their negative
curvature provides room for many peripheral observations while retaining a
compact notion of centrality. Yet constructing a hyperbolic representation
does not itself answer the regression question. An analyst still needs a way
to estimate how covariate associations vary across that representation, to
select neighbourhood sizes, to predict new observations, and to determine
whether the chosen geometry matters.

Lorentz hyperbolic weighted regression (LHWR) addresses this second-stage
problem. For a scalar response \(Y_i\), Euclidean covariates
\(\bx_i\in\R^p\), and a represented location \(Z_i\) on the Lorentz model,
LHWR fits
\begin{equation}
  Y_i=\bx_i^\top\bbeta(Z_i)+\varepsilon_i.
  \label{eq:lhwr-model}
\end{equation}
In Equation~\eqref{eq:lhwr-model}, \(i=1,\ldots,n\) indexes observations,
\(Y_i,\varepsilon_i\in\R\), and
\(\bbeta:\mathbb H_\kappa^d\to\R^p\) is the coefficient field;
\(p\) includes the intercept when one is fitted. The varying coefficient
\(\bbeta(Z_i)\) is estimated by weighted least
squares, with weights determined by Lorentz-hyperbolic distance. Thus the
method does not turn the response or predictors into manifold-valued objects.
It changes only the definition of a local peer.

LHWR belongs to the varying-coefficient framework of
Hastie and Tibshirani \citep{hastietibshirani1993}: regression coefficients are
smooth functions of an indexing variable. Here that index is a hyperbolic
representation, while the response and regression covariates remain Euclidean.
The contribution is not the weighted least-squares identity or a new general
regression family. It is the analysis of intrinsic hyperbolic localization,
curvature-dependent smoothing terms, and representation-induced uncertainty
within an interpretable coefficient model.

The method is connected historically to geographically weighted regression
(GWR) \citep{brunsdon1996,fotheringham2002}, but geography is neither required
nor privileged. It also differs from manifold kernel regression for a scalar
conditional mean \citep{pelletier2006}, tangent-plane local regression
\citep{chengwu2013}, and Fr\'echet regression for metric-space responses
\citep{petersenmuller2019}. Hyperbolic manifold regression
\citep{marconi2020} predicts a hyperbolic-valued response, whereas LHWR uses
hyperbolic locations to index Euclidean coefficient vectors. For estimated
locations, the relevant statistical connection is to regression with
generated covariates \citep{mammen2012}. Our perturbation results identify
the additional weighted-score term and state the conditions under which a
joint first- and second-stage limit transfers to the coefficient estimator;
they do not supply a first-stage limit theorem for every embedding algorithm.

This article is organized around practice. Section~\ref{sec:representation}
explains what the representation means and how it can be supplied. Sections
\ref{sec:method} and \ref{sec:workflow} give the estimator, prediction rule,
bandwidth selection, diagnostics, and a reproducible analysis workflow.
Section~\ref{sec:theory} states the theoretical guarantees needed to interpret
the estimator; Appendices~\ref{app:geometric-definitions}--\ref{app:additional-wdi-diagnostics}
provide the detailed geometry, proofs, and further diagnostics. Sections
\ref{sec:simulation} and \ref{sec:wdi} show what the geometry changes in
simulation and in a country-level application. The emphasis throughout is not
that negative curvature must win, but that representation-defined locality can
be scientifically preferable to global or geographic pooling.

\section{From Representations to Neighbourhoods}
\label{sec:representation}

\subsection{The Lorentz model}

For dimension \(d\) and sectional curvature \(\kappa<0\), let
\(R_\kappa=(-\kappa)^{-1/2}\). The Lorentz inner product on
\(\R^{1,d}\) is
\begin{equation*}
  \innerL{u}{v}=-u_0v_0+\sum_{a=1}^d u_av_a,
\end{equation*}
and the Lorentz model is
\begin{equation*}
  \mathbb H_\kappa^d=
  \{z\in\R^{1,d}:\innerL{z}{z}=-R_\kappa^2,\ z_0>0\}.
\end{equation*}
Its geodesic distance is
\begin{equation}
  d_\kappa(z,w)=R_\kappa\acosh
  \left\{-\frac{\innerL{z}{w}}{R_\kappa^2}\right\}.
  \label{eq:lorentz-distance-general}
\end{equation}
Equation~\eqref{eq:lorentz-distance-general} computes intrinsic distance on
the upper sheet, not the indefinite ambient norm of a coordinate difference.
The radius \(R_\kappa\) fixes the distance scale; all bandwidths expressed as
metric radii use that same scale.

\begin{proposition}[Geodesic metric]
\label{prop:geodesic-metric}
For \(z,w\in\mathbb H_\kappa^d\), the argument of \(\acosh\) in
Equation~\eqref{eq:lorentz-distance-general} is at least one. The distance is
nonnegative, symmetric, zero exactly when \(z=w\), and satisfies the triangle
inequality.
\end{proposition}

The implemented model uses \(d=2\) and \(\kappa=-1\), so
\(\dH(z,w)=\acosh\{-\innerL{z}{w}\}\). The Lorentz and Poincar\'e
models are isometric displays of the same hyperbolic space
\citep{helgason1978,chavel2006}. LHWR uses Lorentz-model coordinates for
distance calculation and the Poincar\'e disk when a bounded two-dimensional
display is easier to read.

\subsection{How coordinates are obtained}

A simple route begins with two scientifically chosen features
\(\psi(x_i)=(\psi_{i1},\psi_{i2})^\top\). The lift
\begin{equation}
  z_i=
  \left(
    \sqrt{1+\psi_{i1}^2+\psi_{i2}^2},
    \psi_{i1},\psi_{i2}
  \right)^\top
  \label{eq:generic-feature-map-Z}
\end{equation}
satisfies \(\innerL{z_i}{z_i}=-1\) and \(z_{i0}>0\). The features may be
economic indicators, clinical summaries, functional scores, shape scores, or
other domain quantities. The lift guarantees a valid Lorentz-model point; it
does not prove that the chosen features are the uniquely correct
representation.

A second route starts from relational data. A graph or similarity matrix can
be embedded by an external method, after which its Poincar\'e or polar
coordinates are converted to the Lorentz model. For example, polar coordinates
\((r_i,\theta_i)\) give
\begin{equation*}
  z_i=(\cosh r_i,\sinh r_i\cos\theta_i,
       \sinh r_i\sin\theta_i)^\top.
\end{equation*}
If the embedding returns \(p_i\) in the Poincar\'e disk, then
\begin{equation*}
  z(p_i)=\frac{(1+\|p_i\|_2^2,2p_{i1},2p_{i2})^\top}
  {1-\|p_i\|_2^2}.
\end{equation*}
Embedding estimation is external to LHWR. The distinction matters: a fixed
feature map can be treated as part of the design, whereas a learned embedding
may contribute first-stage uncertainty and must be reconstructed within each
validation fold to avoid outcome leakage.

The motivation for relational representations comes from tree embeddings
and network geometry \citep{sarkar2011,adcock2013,sala2018,boguna2010}.
Examples include international trade networks \citep{gaulier2010baci,garciaperez2016},
protein interactions \citep{alanislobato2018}, single-cell hierarchies
\citep{klimovskaia2020}, and drug representations \citep{yu2020drug}.
These studies motivate candidate representations, not the correctness of an
LHWR regression for those data. A recent survey describes the broader graph
learning context \citep{yang2025hgl}.

\subsection{Lorentz and tangent neighbourhoods}

Figure~\ref{fig:geometry-comparison} contrasts the two principal local
geometries. LHWR measures neighbourhoods directly on the hyperbolic surface.
Tangent-E maps all observations to one tangent plane at a reference point
\(\mu\) and uses distances there. Because the Lorentz inner product is
positive definite when restricted to \(T_\mu\Htwo\), tangent vectors can be
expressed in a Lorentz-orthonormal basis and compared by ordinary two-dimensional
Euclidean distance. This approximation can be excellent near \(\mu\) or for
smooth broad variation, but it need not preserve pairwise geodesic distances
farther away. The construction is made explicit below so that the comparator
is reproducible \citep{helgason1978,chavel2006,nickel2018}.

For \(\mu\in\mathbb H_\kappa^d\), its tangent space and metric are
\begin{equation}
 T_\mu\mathbb H_\kappa^d=\{v\in\R^{1,d}:\innerL{v}{\mu}=0\},
 \qquad g_\mu(v,w)=\innerL{v}{w},\qquad
 \|v\|_\mu=\sqrt{g_\mu(v,v)}.
 \label{eq:tangent-space-metric}
\end{equation}
Equation~\eqref{eq:tangent-space-metric} supplies a positive-definite norm
only for vectors tangent at \(\mu\). Set \(r=d_\kappa(\mu,z)\). The
logarithmic map sending a represented point to that tangent space is
\begin{equation}
 \log_\mu(z)=
 \begin{cases}
 \displaystyle\frac{r}{R_\kappa\sinh(r/R_\kappa)}
 \left(z+\frac{\innerL{\mu}{z}}{R_\kappa^2}\mu\right),&z\ne\mu,\\[0.6em]
 0,&z=\mu.
 \end{cases}
 \label{eq:tangent-log-map}
\end{equation}
In Equation~\eqref{eq:tangent-log-map}, the scalar multiplier tends to one
as \(r\to0\); the zero case is the continuous extension, not an undefined
division. Moreover, \(\|\log_\mu(z)\|_\mu=r\).

For \(v\in T_\mu\mathbb H_\kappa^d\), let \(s=\|v\|_\mu\). The
exponential map, used below to perturb represented locations, is
\begin{equation}
 \exp_\mu(v)=
 \begin{cases}
 \displaystyle\cosh(s/R_\kappa)\mu+
 R_\kappa\sinh(s/R_\kappa)\frac{v}{s},&s>0,\\[0.4em]
 \mu,&s=0.
 \end{cases}
 \label{eq:exp-map-explicit}
\end{equation}
Equation~\eqref{eq:exp-map-explicit} traces the geodesic with initial
velocity \(v\); orthogonality to \(\mu\) verifies the hyperboloid
constraint. The ratio \(R_\kappa\sinh(s/R_\kappa)/s\) tends to one.
The inverse relation with Equation~\eqref{eq:tangent-log-map} is verified
in Appendix~\ref{app:geometric-definitions}.

If
\(E_1,\ldots,E_d\) is a Lorentz-orthonormal basis of the tangent space,
write \(q_a(z)=g_\mu\{\log_\mu(z),E_a\}\) and
\(q(z)=(q_1(z),\ldots,q_d(z))^\top\). Tangent-E uses
\begin{equation}
 d_{T,\mu}(z_i,z_j)=
 \|\log_\mu(z_i)-\log_\mu(z_j)\|_\mu
 =\|q(z_i)-q(z_j)\|_2.
 \label{eq:tangent-distance}
\end{equation}
Equation~\eqref{eq:tangent-distance} is independent of the chosen
orthonormal basis, but generally depends on \(\mu\). Radial distances from
\(\mu\) are preserved; distances between arbitrary pairs need not be.
A quantitative bound for this approximation is given in Appendix~\ref{app:geometric-definitions}.

Unless a reference is supplied, the implementation uses the normalized
extrinsic Lorentz mean of the \(n_{\rm tr}\) training locations:
\begin{equation}
 \bar z_{\rm tr}=n_{\rm tr}^{-1}\sum_{i\in\mathcal I_{\rm tr}}z_i,
 \qquad
 \mu=\frac{R_\kappa\bar z_{\rm tr}}
 {\sqrt{-\innerL{\bar z_{\rm tr}}{\bar z_{\rm tr}}}}.
 \label{eq:tangent-reference}
\end{equation}
Here \(\mathcal I_{\rm tr}\) is the training index set. The denominator in
Equation~\eqref{eq:tangent-reference} is positive for upper-sheet inputs;
Appendix~\ref{app:geometric-definitions} proves this and verifies the
inverse relation between the logarithmic and exponential maps.
This reference is not defined as the minimizer of squared geodesic distances
(the intrinsic Fr\'echet mean). The same training reference is retained when
mapping validation or future observations. In outer validation it is computed
anew from each training fold. Internal bandwidth scoring holds the resulting
training distance matrix fixed; it is not a second re-estimation of the
reference for every omitted response. These distinctions describe the existing
protocol, not a change to the analysis. The theory below concerns intrinsic
LHWR weights; it does not automatically give a separate limit theorem for
Tangent-E with a random reference.

At unit curvature, the log map uses tolerance \(10^{-10}\) for the computed
geodesic distance \(r=d_{-1}(\mu,z)\): if \(r\leq10^{-10}\), tangent
coordinates are set to zero; otherwise the logarithmic map is evaluated.

This numerical convention is distinct from the exact map and distance bound
in Appendix~\ref{app:geometric-definitions}. Pairwise tangent distances are
evaluated from coordinate differences. In particular,
identical mapped points have exactly zero distance, whereas close distinct
coordinates are not collapsed by an additional pairwise-distance threshold.

\begin{figure}[!t]
  \centering
  \includegraphics[width=0.98\textwidth]{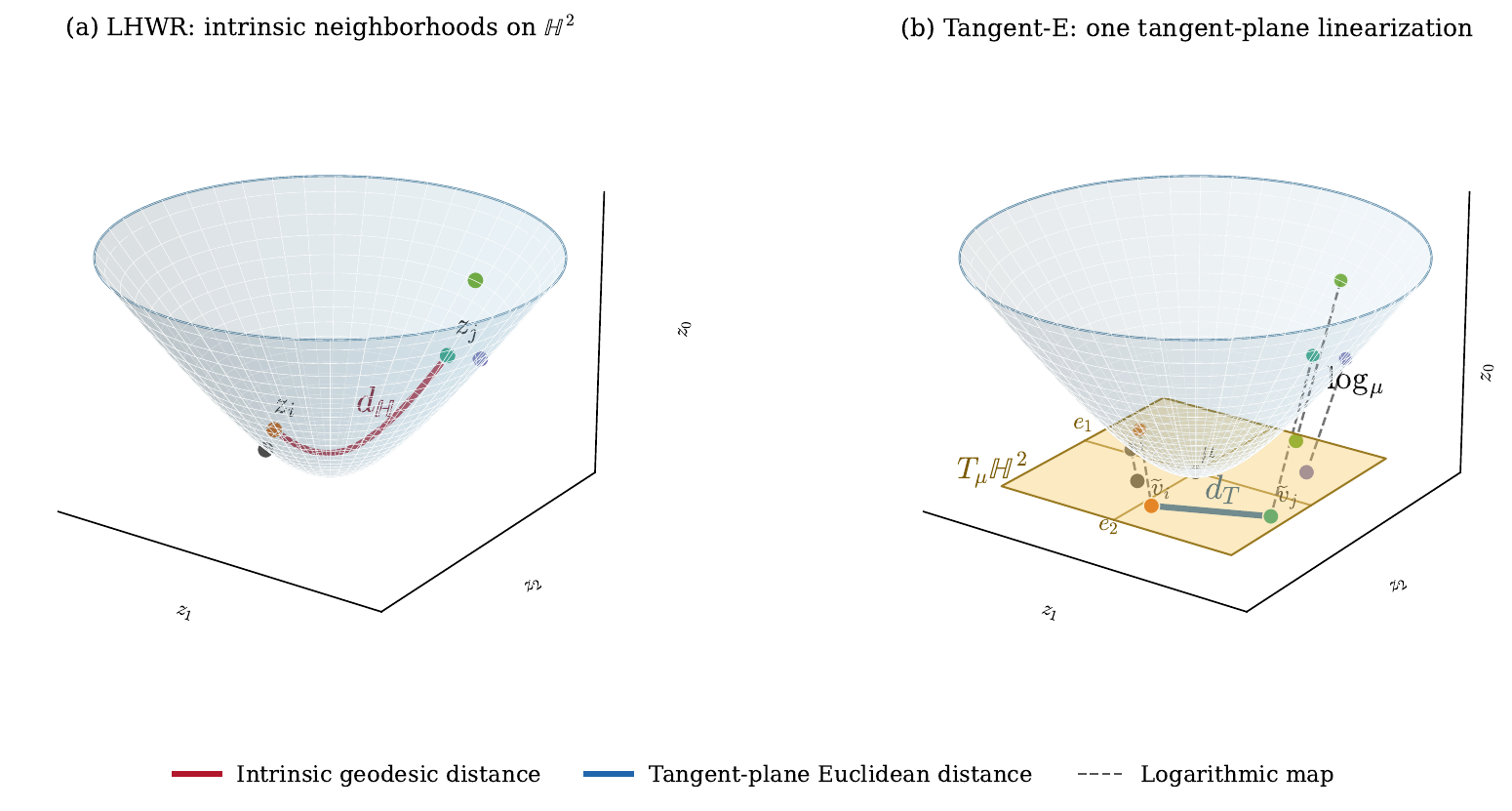}
  \caption{Two ways to define locality from the same Lorentz-model
  representation. The left panel uses geodesic distance directly on
  \(\mathbb H^2\). The right panel maps observations to the single tangent
  plane \(T_\mu\mathbb H^2\) before measuring distance.}
  \label{fig:geometry-comparison}
\end{figure}

\section{The LHWR Method}
\label{sec:method}

Let \(\bX=(\bx_1^\top,\ldots,\bx_n^\top)^\top\),
\(\by=(Y_1,\ldots,Y_n)^\top\), and
\(\bZ=(Z_1^\top,\ldots,Z_n^\top)^\top\). At a target \(z\), define
\begin{equation*}
  w_i(z;h)=K\{d_\kappa(z,Z_i)/h\},\qquad
  \bW(z;h)=\diag\{w_1(z;h),\ldots,w_n(z;h)\}.
\end{equation*}
The estimator is
\begin{equation}
  \widehat\bbeta_h(z)=
  \{\bX^\top\bW(z;h)\bX\}^{-1}
  \bX^\top\bW(z;h)\by,
  \label{eq:lhwr-estimator}
\end{equation}
provided the weighted design has full column rank. Here \(\bX\) is
\(n\times p\), \(\by\) is \(n\times1\), \(\bW\) is \(n\times n\),
\(h>0\), and \(K:[0,\infty)\to[0,\infty)\) is a specified kernel.
Equation~\eqref{eq:lhwr-estimator} estimates a \(p\)-vector at each target;
it does not require \(\bbeta\) to be constant over the representation.

\begin{assumption}[Finite-sample local design]
\label{ass:local-design}
For every reported target and bandwidth, the nonnegative weights are fixed
and \(\bW^{1/2}\bX\) has column rank \(p\).
\end{assumption}

\begin{proposition}[Existence, uniqueness, and weighted optimality]
\label{prop:local-wls}
Under Assumption~\ref{ass:local-design}, the estimator in
Equation~\eqref{eq:lhwr-estimator} exists and uniquely minimizes
\begin{equation}
 Q_z(\theta)=\sum_{i=1}^n w_i(z;h)
               \{Y_i-\bx_i^\top\theta\}^2,\qquad\theta\in\R^p.
 \label{eq:local-wls-objective}
\end{equation}
\end{proposition}
Equation~\eqref{eq:local-wls-objective} defines optimality for the selected
local loss. It is not a claim that LHWR minimizes prediction risk among
different geometries. Singular local designs must be flagged; a numerical
generalized-inverse fallback does not establish uniqueness in the proposition.
The experiments use the bisquare kernel
\(K(u)=(1-u^2)^2\mathbf1(|u|<1)\). Gaussian, exponential, boxcar, and
tricube kernels are available alternatives
\citep{gollini2013gwmodel,yuzbasi2026shrinkage}.

The prediction for a represented future observation
\((\bx_\star,z_\star)\) is
\begin{equation}
  \widehat y_\star=
  \bx_\star^\top\widehat\bbeta(z_\star),
  \label{eq:out-of-sample-prediction}
\end{equation}
where weights compare \(z_\star\) only with training coordinates. A fixed
bandwidth uses one metric radius. An adaptive bandwidth uses the distance to
the \(k\)th nearest training observation, allowing dense and sparse parts of
the representation to contain comparable neighbourhood mass. The package
orders strictly positive finite distances: coincident locations are omitted
from radius selection but retain their kernel weight at distance zero. With a
bisquare kernel the observation at the radius has zero weight, so \(k\) is a
radius-selection parameter, not a guarantee of exactly \(k\) positive weights.
Equation~\eqref{eq:out-of-sample-prediction} uses the training response only;
the target covariates enter through the final inner product.

For observed targets, stacking the fitted-value rows gives a smoother matrix
\(\widehat\by=\bS\by\). Let
\(\RSS=\|\by-\widehat\by\|_2^2\). The package selects candidate
bandwidths by leave-one-out cross-validation or by the corrected Akaike
information criterion
\begin{equation}
  \AICc=n\log(\RSS/n)+n\log(2\pi)+n+
  \frac{2n\{\tr(\bS)+1\}}{n-\tr(\bS)-2}.
  \label{eq:bandwidth-aicc}
\end{equation}
Equation~\eqref{eq:bandwidth-aicc} scores admissible fits with
\(\RSS>0\) and \(n-\tr(\bS)-2>0\); \(\tr(\bS)\) measures the
conditional smoother complexity. For the same fixed distance system, define
the leave-one-out criterion
\begin{equation}
 \CV(h)=\sum_{i=1}^n
 \{Y_i-\bx_i^\top\widehat\bbeta_h^{(-i)}(Z_i)\}^2.
 \label{eq:bandwidth-cv}
\end{equation}
In Equation~\eqref{eq:bandwidth-cv}, superscript \((-i)\) excludes response
\(i\) from the weighted fit, using the candidate's specified radius rule.
This conditional score is distinct from outer validation of a learned
representation and bandwidth-selection procedure.
The same kernel, candidate set, and selection criterion should be used when
comparing distance geometries.

Algorithm~\ref{alg:lhwr} summarizes the complete LHWR calibration procedure,
using the intrinsic geometry at \(d=2\), \(\kappa=-1\). In the algorithm,
\(\mathcal B\) is the search set, \(b\) a candidate, \(b_i(b)\) its local
metric radius, and \(\widehat b_i=b_i(\widehat b)\) the selected radius.
The residual vector is \(\be=\by-\widehat\by\); the learned rule is
\(\widehat m_n(\bx,z)=\bx^\top\widehat\bbeta_n(z)\), where subscript
\(n\) emphasizes use of the training sample. The positive-distance and
admissibility conventions above apply to the algorithm's nearest-neighbour
search and candidate fits. Its provisional and final coefficient equations
are instances of Equation~\eqref{eq:lhwr-estimator}, and its prediction uses
Equation~\eqref{eq:out-of-sample-prediction}. The implementation solves the
normal equations; the displayed inverse is not a requirement to form a dense
inverse explicitly. Tangent-E is a separate geometry ablation defined by
Equations~\eqref{eq:tangent-log-map}--\eqref{eq:tangent-reference}.
\begin{algorithm}[!t]
    \renewcommand{\baselinestretch}{1}\fontsize{10}{12}\selectfont
    \ResetInOut{Output}
    \SetKwInOut{Input}{Input}
    \SetKwInOut{Output}{Output}
    \Input{Response vector $\by\in\R^n$, design matrix $\bX\in\R^{n\times p}$,
    Lorentz-model coordinate matrix $\bZ=(z_1^\top,\ldots,z_n^\top)^\top$ with
    $z_i\in\Htwo$, kernel $K$, bandwidth search set $\mathcal B$, bandwidth
    type, criterion $\mathcal C\in\{\mathrm{AICc},\mathrm{CV}\}$, and optional
    new pair $(\bx_\star,z_\star)$}
    \Output{Local coefficient surfaces $\{\widehat\bbeta(z_i)\}_{i=1}^n$,
    fitted values $\widehat\by$, residuals $\be$, selected bandwidth
    $\widehat b$, learned prediction rule $\widehat m_n$, optional prediction
    $\widehat y_\star$, and diagnostic summaries}

    \For{$i\gets 1$ \KwTo $n$}{
        Check $\langle z_i,z_i\rangle_L\approx -1$ and $z_{i0}>0$\;
    }
    If coordinates are obtained from a feature map, Poincar\'e embedding,
    network embedding, or similarity representation, convert or project them
    to the upper sheet of $\Htwo$ before fitting\;

    \ForEach{required pair $(i,j)$}{
        Compute the Lorentz-hyperbolic distance
        $\dH(z_i,z_j)=\operatorname{arcosh}
        \{-\langle z_i,z_j\rangle_L\}$\;
    }
    For exact kNN fitting, retain only the nearest-neighbor indices and
    distances required by the compact kernel and adaptive bandwidth\;

    \ForEach{candidate bandwidth $b\in\mathcal B$}{
        Construct local bandwidths $b_i(b)$\;
        \If{adaptive kNN bandwidth is used}{
            Set $b_i(b)=d_i(k)$, the Lorentz-hyperbolic distance from $z_i$ to its
            $k$th nearest neighbor\;
        }
        \For{$i\gets 1$ \KwTo $n$}{
            Set $w_{ij}^{(b)}=K\{\dH(z_i,z_j)/b_i(b)\}$,
            $\bW_i^{(b)}=\operatorname{diag}(w_{i1}^{(b)},\ldots,w_{in}^{(b)})$,
            and $\bG_i^{(b)}=\bX^\top\bW_i^{(b)}\bX$\;
            Compute provisional
            $\widehat\bbeta_b(z_i)=(\bG_i^{(b)})^{-1}\bX^\top\bW_i^{(b)}\by$\;
            Compute $\widehat y_{i,b}=\bx_i^\top\widehat\bbeta_b(z_i)$\;
        }
        Evaluate the selection score $\mathcal C(b)$ using AICc or leave-one-out CV\;
    }
    Select $\widehat b=\argmin_{b\in\mathcal B}\mathcal C(b)$\;

    \For{$i\gets 1$ \KwTo $n$}{
        Construct final weights
        $w_{ij}=K\{\dH(z_i,z_j)/\widehat b_i\}$,
        $\bW_i=\operatorname{diag}(w_{i1},\ldots,w_{in})$,
        and $\bG_i=\bX^\top\bW_i\bX$\;
        Compute the local coefficient estimate
        $\widehat\bbeta(z_i)=\bG_i^{-1}\bX^\top\bW_i\by$\;
        \If{$\bG_i$ is ill-conditioned}{
            Record local condition diagnostics and flag the fit as numerically
            nonregular. If direct inversion fails, the numerical implementation
            uses a generalized inverse as a fallback; such a solution may be
            nonunique and lies outside the regular full-rank theory\;
        }
        Compute $\widehat y_i=\bx_i^\top\widehat\bbeta(z_i)$ and store the
        residual component $e_i$\;
    }

    Compute requested diagnostics: \RSS, \AIC, \AICc, \CV, $\operatorname{tr}(\bS)$,
    effective degrees of freedom, local standard errors, local $t$-values,
    local collinearity diagnostics, conditional F-tests, and residual
    autocorrelation diagnostics under hyperbolic neighborhoods\;
    \If{$ (\bx_\star,z_\star)$ is supplied}{
        Construct $\bW_\star$ from training coordinates and the calibrated
        bandwidth; compute $\widehat\bbeta_n(z_\star)$ and
        $\widehat y_\star=\bx_\star^\top\widehat\bbeta_n(z_\star)$\;
    }
    \textbf{return} $\{\widehat\bbeta(z_i)\}_{i=1}^n$, $\widehat\by$, $\be$,
    $\widehat b$, $\widehat m_n$, optional $\widehat y_\star$, and diagnostics\;
    \caption{Lorentz Hyperbolic Weighted Regression calibration.}
    \label{alg:lhwr}
\end{algorithm}

\section{Diagnostics and a Practical Workflow}
\label{sec:workflow}

\subsection{What should be checked}

The first check is whether the local design is estimable. At location \(i\),
the effective sample size, number of positive weights, condition number, and
smallest eigenvalue of \(\bX^\top\bW_i\bX\) reveal sparse or nearly singular
neighbourhoods. Local variance inflation factors identify predictor-specific
collinearity \citep{wheeler2005}. These diagnostics concern the weighted
design, not the curvature itself: a valid Lorentz representation can still
produce an unstable regression if too few locally distinct observations are
available.

Conditional coefficient covariance follows from the linear representation
\(\widehat\bbeta(Z_i)=\bA_i\by\), where
\begin{equation*}
  \bA_i=(\bX^\top\bW_i\bX)^{-1}\bX^\top\bW_i.
\end{equation*}
Let \(\bm m=\E(\by\mid\bX,\bZ)\), with entries
\(m_j=\bx_j^\top\bbeta(Z_j)\), and suppose
\(\E(\bepsilon\mid\bX,\bZ)=0\) and
\(\Var(\bepsilon\mid\bX,\bZ)=\bm\Sigma\), an \(n\times n\)
covariance matrix. With weights fixed independently of the response errors,
the exact conditional identities are
\begin{equation}
\begin{aligned}
 \E\{\widehat\bbeta(Z_i)\mid\bX,\bZ\}&=\bA_i\bm m,\\
 \Var\{\widehat\bbeta(Z_i)\mid\bX,\bZ\}&=\bA_i\bm\Sigma\bA_i^\top,\\
 \operatorname{Bias}_i&=\bA_i\{\bm m-\bX\bbeta(Z_i)\}.
\end{aligned}
\label{eq:conditional-moments}
\end{equation}
Equation~\eqref{eq:conditional-moments} separates sampling variation from
the smoothing bias of a varying field. Under homoskedastic uncorrelated errors,
\(\bm\Sigma=\sigma^2\bI_n\), giving the plug-in covariance
\(\widehat\sigma^2\bA_i\bA_i^\top\). Simply conditioning on a
response-selected bandwidth does not preserve the assumed error law; reported
plug-in covariances do not incorporate that selection uncertainty.
Pointwise coefficient maps are useful
for exploration, but displaying many local tests creates multiplicity and the
bandwidth was selected from the data. Raw and Benjamini--Hochberg-adjusted
classifications should therefore be distinguished explicitly.

Global and local residual autocorrelation can be summarized using the same
Lorentz-hyperbolic neighbourhoods that fitted the model. For row-standardized
weights \(w_{ij}^{(H)}\) with zero diagonal, the residual Moran statistic is
\begin{equation*}
  I_H=\frac{n}{S_0}
  \frac{\sum_{i,j}w_{ij}^{(H)}(e_i-\bar e)(e_j-\bar e)}
       {\sum_i(e_i-\bar e)^2},
  \qquad S_0=\sum_{i,j}w_{ij}^{(H)}.
\end{equation*}
The corresponding local indicators identify high--high, low--low,
high--low, and low--high residual patterns
\citep{moran1950,anselin1995}. Because smoothing makes fitted residuals
nonexchangeable, ordinary residual permutations are descriptive post-fit
checks rather than exact tests. Moment-matched F diagnostics for global and
coefficient-specific structure have the same conditional limitation
\citep{leung2000}; their definitions and calibration are given in
Appendices~\ref{app:secondary-diagnostics} and~\ref{app:diagnostic-calibration}.

The representation should encode a scientifically defensible notion of
similarity. Comparisons with a global regression baseline and alternative
neighbourhood geometries should use the same observations and covariates,
with a common kernel and bandwidth-selection protocol for the local methods.
For predictive assessment, any learned representation must be re-estimated
and local bandwidths reselected using only the training observations in each
validation split. Local coefficient maps should be interpreted alongside the
design diagnostics and with the uncertainty convention stated explicitly.
Residual Moran and LISA summaries remain descriptive post-fit diagnostics.
The resulting local associations concern the chosen representation and do
not, by themselves, establish causal effects.

\section{What the Theory Guarantees}
\label{sec:theory}

This section records the results that guide bandwidth choice and uncertainty
interpretation. Detailed proofs appear in Appendices~\ref{app:geometric-definitions}--\ref{app:generated-representation-proofs}.

Suppose \((Y_i,\bx_i,Z_i)\) are independent and identically distributed and satisfy
\begin{equation}
  Y_i=\bx_i^\top\bbeta(Z_i)+\varepsilon_i,
  \qquad \E(\varepsilon_i\mid\bx_i,Z_i)=0.
  \label{eq:asymptotic-model}
\end{equation}
In Equation~\eqref{eq:asymptotic-model}, \(p\) is fixed and the mean-zero
condition permits heteroskedasticity. Let \(f\) be the density of \(Z_i\)
with respect to hyperbolic volume \(dV_\kappa\),
and let \(\bm M(z)=\E(\bx_i\bx_i^\top\mid Z_i=z)\). Define
\begin{equation}
  \bm Q(z)=f(z)\bm M(z).
  \label{eq:asymptotic-Q}
\end{equation}
Let
\(\bm\Omega(z)=\E(\varepsilon_i^2\bx_i\bx_i^\top\mid Z_i=z)\).
Equation~\eqref{eq:asymptotic-Q} combines density and the conditional design
second moment; \(\bm Q,\bm M,\bm\Omega\) are \(p\times p\) matrices.
For a radial kernel, with all integrals over \(\R^d\), define
\(\mu_0=\int K(\|u\|)du\),
\(\mu_2=\int u_1^2K(\|u\|)du\),
\(\nu_0=\int K(\|u\|)^2du\), and
\(\nu_2=\int u_1^2K(\|u\|)^2du\).

\begin{assumption}[Pointwise regularity]
\label{ass:pointwise-asymptotics}
At a fixed interior target \(z\), the fields \(f\), \(\bm M\),
\(\bm\Omega\), and \(\bbeta\) have four bounded covariant derivatives;
\(f(z)>0\); and \(\bm M(z)\) and \(\bm\Omega(z)\) are positive definite.
For some \(\delta>0\), the conditional moments of \(\|\bx_i\|^{4+\delta}\)
and \(|\varepsilon_i|^{4+\delta}\|\bx_i\|^{4+\delta}\), given \(Z_i=q\),
are uniformly bounded near \(z\). The kernel is nonnegative, bounded,
Lipschitz, radial, and supported on \([0,1]\), with
\(0<\mu_0,\mu_2,\nu_0<\infty\). The deterministic bandwidth satisfies
\(h\to0\) and \(nh^d\to\infty\).
\end{assumption}

Let \(\nabla\) be the Levi--Civita covariant derivative and
\(e_1,\ldots,e_d\) an orthonormal frame at the target, extended normally
there. The Laplace--Beltrami convention is
\(\Delta_\kappa=\operatorname{div}_g\operatorname{grad}_g\), applied
entrywise to vector or matrix fields; it reduces to \(\sum_a\partial_a^2\)
in flat coordinates. Its coordinate definition is given in Appendix~\ref{app:geometric-definitions}. With these conventions, define the leading smoothing-bias vector
\begin{equation}
  \bm B(z)=\frac{\mu_2}{2\mu_0}\bm Q(z)^{-1}
  \left[
    \Delta_\kappa\{\bm Q\bbeta\}(z)
    -\{\Delta_\kappa\bm Q(z)\}\bbeta(z)
  \right],
  \label{eq:asymptotic-bias-constant}
\end{equation}
equivalently,
\begin{equation}
  \bm B(z)=\frac{\mu_2}{2\mu_0}
  \left[
    \Delta_\kappa\bbeta(z)+2\bm Q(z)^{-1}
    \sum_{a=1}^d\{\nabla_{e_a}\bm Q(z)\}
    \{\nabla_{e_a}\bbeta(z)\}
  \right].
  \label{eq:asymptotic-bias-equivalent}
\end{equation}
Equations~\eqref{eq:asymptotic-bias-constant} and
\eqref{eq:asymptotic-bias-equivalent} are linked by the covariant product
rule: the second form separates curvature of the coefficient field from
gradients of density and design. The leading covariance constant is
\begin{equation}
  \bm V(z)=\frac{\nu_0}{\mu_0^2f(z)}
  \bm M(z)^{-1}\bm\Omega(z)\bm M(z)^{-1}.
  \label{eq:asymptotic-variance-constant}
\end{equation}
Equation~\eqref{eq:asymptotic-variance-constant} gives a \(p\times p\)
matrix; unlike \(\bm B(z)\), it depends on the conditional error second
moment. Define \(K_{h,z}(q)=K\{d_\kappa(z,q)/h\}\) and
\begin{equation}
\begin{aligned}
 \bm\Gamma_h(z)&=\E\{K_{h,z}(Z)\bx\bx^\top\},&
 \bm g_h(z)&=\E\{K_{h,z}(Z)\bx Y\},\\
 \widehat{\bm\Gamma}_{n,h}(z)&=n^{-1}\sum_iw_i(z;h)\bx_i\bx_i^\top,&
 \widehat{\bm g}_{n,h}(z)&=n^{-1}\sum_iw_i(z;h)\bx_iY_i.
\end{aligned}
\label{eq:local-population-moments}
\end{equation}
The matrices and vectors in Equation~\eqref{eq:local-population-moments}
have dimensions \(p\times p\) and \(p\times1\), respectively. Kernel
weights are intentionally unnormalized: a common factor \(h^{-d}\) cancels
from the local estimator. The population target and sample estimator are

\begin{equation}
\begin{aligned}
  \bbeta_h(z)&=\bm\Gamma_h(z)^{-1}\bm g_h(z),\\
  \widehat\bbeta_h(z)&=\widehat{\bm\Gamma}_{n,h}(z)^{-1}
                         \widehat{\bm g}_{n,h}(z).
\end{aligned}
\label{eq:population-and-sample-estimator}
\end{equation}

Equation~\eqref{eq:population-and-sample-estimator} distinguishes the
population smoothing target \(\bbeta_h\) from both \(\bbeta\) and the
random estimate. The sample estimator is defined when its local Gram matrix is nonsingular,
an event whose probability tends to one. For unconditional moment statements
only, assign a fixed bounded vector on the complementary event; the moment
conditions below apply to this completion. This convention does not change
the fitted estimator on admissible local designs.

\begin{theorem}[Pointwise behavior]
\label{thm:pointwise-consistency}
Under Assumption~\ref{ass:pointwise-asymptotics},
\begin{equation}
  \widehat\bbeta_h(z)-\bbeta(z)=
  h^2\bm B(z)+O(h^4)+O_p\{(nh^d)^{-1/2}\},
  \label{eq:pointwise-rate}
\end{equation}
and therefore \(\widehat\bbeta_h(z)\xrightarrow{p}\bbeta(z)\).
\end{theorem}
Equation~\eqref{eq:pointwise-rate} separates a second-order smoothing term
from local sampling noise. Both vanish under the bandwidth conditions of
Assumption~\ref{ass:pointwise-asymptotics}.

\begin{theorem}[Population bias and sampling variance]
\label{thm:bias-variance}
Under Assumption~\ref{ass:pointwise-asymptotics}, the population expansion
below holds. The unconditional variance statement additionally
requires uniform integrability of
\(nh^d\|\widehat\bbeta_h(z)-\bbeta_h(z)\|^2\), including the completed
singular-design event. This extra condition is not needed for the
probability or distribution limits. Then
\begin{align}
  \bbeta_h(z)-\bbeta(z)&=h^2\bm B(z)+O(h^4),
  \label{eq:population-bias-expansion}\\
  \Var\{\widehat\bbeta_h(z)\}&=
  \frac{1}{nh^d}\bm V(z)+o\{(nh^d)^{-1}\}.
  \label{eq:variance-expansion}
\end{align}
\end{theorem}
Equation~\eqref{eq:population-bias-expansion} describes smoothing bias, not
an unqualified identity for \(\E\widehat\bbeta_h-\bbeta\).
Equation~\eqref{eq:variance-expansion} requires the stated moment control
because rare nearly singular local designs can affect unconditional moments.

\begin{theorem}[Pointwise asymptotic normality]
\label{thm:asymptotic-normality}
Under Assumption~\ref{ass:pointwise-asymptotics},
\begin{equation}
  \sqrt{nh^d}\{\widehat\bbeta_h(z)-\bbeta_h(z)\}
  \Rightarrow N_p\{\bm0,\bm V(z)\}.
  \label{eq:clt-population-target}
\end{equation}
If \(nh^{d+8}\to0\), then
\begin{equation}
  \sqrt{nh^d}\{\widehat\bbeta_h(z)-\bbeta(z)-h^2\bm B(z)\}
  \Rightarrow N_p\{\bm0,\bm V(z)\}.
  \label{eq:clt-bias-centred}
\end{equation}
If the stronger undersmoothing condition \(nh^{d+4}\to0\) holds, then
\begin{equation}
 \sqrt{nh^d}\{\widehat\bbeta_h(z)-\bbeta(z)\}
 \Rightarrow N_p\{\bm0,\bm V(z)\}.
 \label{eq:clt-undersmoothed}
\end{equation}
\end{theorem}
Equations~\eqref{eq:clt-population-target}, \eqref{eq:clt-bias-centred},
and \eqref{eq:clt-undersmoothed} distinguish centring at the population
target, explicit bias centring, and making the leading bias negligible.
None of these distribution limits alone guarantees moment convergence.

The population smoothing bias and leading covariance in
Theorem~\ref{thm:bias-variance} give the asymptotic mean squared error criterion
\begin{equation}
  \operatorname{AMSE}_z(h)=h^4\|\bm B(z)\|_2^2+
  \frac{\tr\{\bm V(z)\}}{nh^d},
  \label{eq:amse}
\end{equation}
and, when \(\bm B(z)\ne0\),
\begin{equation}
  h_{\mathrm{AMSE}}(z)=
  \left[\frac{d\,\tr\{\bm V(z)\}}
  {4n\|\bm B(z)\|_2^2}\right]^{1/(d+4)}.
  \label{eq:optimal-bandwidth}
\end{equation}
Equation~\eqref{eq:amse} balances squared smoothing bias against sampling
variance. Minimizing it gives Equation~\eqref{eq:optimal-bandwidth}, not an
automatic justification of a data-selected bandwidth's inferential law.
Thus \(h\asymp n^{-1/(d+4)}\); in two dimensions,
\(h\asymp n^{-1/6}\) and the optimal mean squared error is of order
\(n^{-2/3}\).

\subsection{Curvature and adaptive neighbourhoods}

Curvature enters through geodesic normal coordinates, the
Laplace--Beltrami derivatives, and the hyperbolic volume element. For constant
sectional curvature the contraction chain is
\begin{equation}
 \operatorname{Ric}_\kappa=(d-1)\kappa g,
 \qquad \operatorname{Scal}_\kappa
 =\operatorname{tr}_g\operatorname{Ric}_\kappa=d(d-1)\kappa.
 \label{eq:curvature-contractions}
\end{equation}
Equation~\eqref{eq:curvature-contractions} links sectional curvature to the
Ricci bilinear form and then to scalar curvature. Appendix~\ref{app:geometric-definitions} gives the polar metric and exact ball-volume integral: for \(d\geq2\),
large-radius volume grows proportionally to
\(\exp\{(d-1)r/R_\kappa\}\). A regular rooted tree with integer branching factor
\(b_{\rm tree}\geq2\)
has \(b_{\rm tree}^{\ell}\) nodes at level \(\ell=0,1,\ldots\); this
exponential increase parallels the volume capacity of hyperbolic space
\citep{krioukov2010}. This motivates hierarchical neighbourhoods;
it is not a regression-risk comparison. For the
constant-curvature model used here, the normal-coordinate Jacobian satisfies
\(J_z(v)=1-\operatorname{Ric}_\kappa(v,v)/6+O(\|v\|^4)\). Because the kernel is
radial, angular integration contracts the Ricci tensor to its trace, so that
\(\int K(\|u\|)\operatorname{Ric}_\kappa(u,u)\,du
=\mu_2\operatorname{Scal}_\kappa\). This explains why the resulting expansion
is expressed in terms of scalar curvature. For a scalar, vector, or matrix
field \(\bm F\) with four bounded covariant derivatives near an interior
target and a kernel satisfying Assumption~\ref{ass:pointwise-asymptotics},
\begin{equation}
\begin{aligned}
 &\int K\{d_\kappa(z,q)/h\}\bm F(q)\,dV_\kappa(q)\\
 &\quad=h^d\left[\mu_0\bm F(z)+h^2\mu_2
 \left\{\tfrac12\Delta_\kappa\bm F(z)
 -\tfrac16\operatorname{Scal}_\kappa\bm F(z)\right\}+O(h^4)\right].
\end{aligned}
\label{eq:curved-kernel-moment}
\end{equation}
Equation~\eqref{eq:curved-kernel-moment} is the radial-moment lemma proved
in Appendix~\ref{app:normal-coordinate-expansion}. It is applied entrywise to the local Gram and score
moments, explaining where geometry enters their statistical expansions.

\begin{proposition}[Curvature contribution]
\label{prop:curvature-contribution}
Under Assumption~\ref{ass:pointwise-asymptotics}, for constant sectional curvature \(\kappa\),
\(\operatorname{Scal}_\kappa=d(d-1)\kappa\), and
\begin{equation}
  \operatorname{Vol}_\kappa\{B(z,h)\}=v_dh^d
  \left[1-\frac{\operatorname{Scal}_\kappa}{6(d+2)}h^2+O(h^4)\right].
  \label{eq:small-ball-volume}
\end{equation}
Here \(v_d\) is the volume of the Euclidean unit ball in \(\R^d\), and
\(B(z,h)\) denotes a geodesic ball. Equation~\eqref{eq:small-ball-volume}
quantifies its departure from flat volume as \(h\to0\).
The scalar-curvature term in Equation~\eqref{eq:curved-kernel-moment}
cancels from the leading local-constant coefficient bias, but remains in
neighbourhood volume and higher-order variance. To isolate its variance
contribution, suppose \(f\), \(\bm M\), \(\bm\Omega\), and \(\bbeta\) are
constant near \(z\), local covariates and conditional error variance are
bounded, and the scaled covariance-sandwich moment condition in
Appendix~\ref{app:curvature-proofs} holds. If also \(nh^{d+2}\to\infty\), then
\begin{equation}
  \Var\{\widehat\bbeta_h(z)\}=\frac{\bm V(z)}{nh^d}
  \left[1+\frac{\operatorname{Scal}_\kappa}{6}
  \left(2\frac{\mu_2}{\mu_0}-\frac{\nu_2}{\nu_0}\right)h^2\right]
  +o\{h^2/(nh^d)\}.
  \label{eq:curvature-variance-correction}
\end{equation}
These additional conditions control the random-inverse contribution;
Equation~\eqref{eq:curvature-variance-correction} is a curvature-only
expansion, not the full second-order variance
formula for varying design or coefficient fields.
\end{proposition}

\begin{corollary}[Adaptive nearest-neighbour rate]
\label{cor:adaptive-knn}
Assume the model and local conditions of
Assumption~\ref{ass:pointwise-asymptotics}, replacing its deterministic
bandwidth condition by \(k\to\infty\), \(k/n\to0\). At a fixed interior
target with \(f(z)>0\), the \(k\)th-neighbour radius has the expansion
\begin{equation}
\begin{aligned}
 H_{n,k}(z)=h_{n,k}^{(0)}(z)\Bigg[1+
 \left\{\frac{\operatorname{Scal}_\kappa}{6d(d+2)}
 -\frac{\Delta_\kappa f(z)}{2d(d+2)f(z)}\right\}
 \{h_{n,k}^{(0)}(z)\}^2\\
 +O\{\{h_{n,k}^{(0)}(z)\}^4\}+O_p(k^{-1/2})\Bigg],
\end{aligned}
\label{eq:knn-radius-expansion}
\end{equation}
where \(h_{n,k}^{(0)}(z)=\{k/(nf(z)v_d)\}^{1/d}\). Let
\(\bar h_{n,k}=F_z^{-1}(k/n)\), where
\(F_z(r)=\Pr\{d_\kappa(z,Z_i)\leq r\}\), be the corresponding population
radius. The smoothing bias at \(\bar h_{n,k}\) is
\(O\{(k/n)^{2/d}\}\), and
\(\widehat\bbeta_{H_{n,k}}-\bbeta_{\bar h_{n,k}}=O_p(k^{-1/2})\).
For unconditional variance statements, additionally require uniform
integrability of
\(k\|\widehat\bbeta_{H_{n,k}}-\bbeta_{\bar h_{n,k}}\|^2\), including
the bounded singular-design completion. Then the variance is \(O(k^{-1})\).
Balancing squared smoothing bias with that variance gives
\(k\asymp n^{4/(d+4)}\) when the leading bias is nonzero.
\end{corollary}
Equation~\eqref{eq:knn-radius-expansion} separates the density/curvature
correction from order-statistic noise. Its transfer to an estimator uses the
conditional order-statistic argument in Appendix~\ref{app:curvature-proofs}, not substitution of
a random bandwidth into a pointwise deterministic-bandwidth theorem. The
continuous location density makes distance ties and coincident target points
null events in this result.

\begin{theorem}[Flat-curvature limit]
\label{thm:flat-limit}
Let \(\kappa_m<0\) tend to zero with fixed \(d,p,K\). Identify neighbourhoods
of targets \(z_m\) using one bounded normal-coordinate domain. Suppose the
pulled-back fields \(f_m,\bm M_m,\bm\Omega_m,\bbeta_m\) converge in
\(C^2\) to their Euclidean counterparts, have uniformly bounded fourth
covariant derivatives and the moments in
Assumption~\ref{ass:pointwise-asymptotics}, and their densities and relevant
minimum eigenvalues are uniformly bounded away from zero near the targets.
Let \(n_m\to\infty\), \(h_m\to0\), \(n_mh_m^d\to\infty\), and
\(|\kappa_m|h_m^2\to0\). Then the normalized kernel-moment expansions
and constants \(\bm B_m,\bm V_m\) converge to their Euclidean
local-regression counterparts. With \(n_mh_m^{d+8}\to0\), the bias-centred
normal limit also converges. For unconditional covariance convergence, impose
the uniform-integrability condition of Theorem~\ref{thm:bias-variance}
along this sequence. For adaptive radii, impose the conditions of
Corollary~\ref{cor:adaptive-knn} uniformly in \(m\),
\(k_m\to\infty\), \(k_m/n_m\to0\), and
\(|\kappa_m|(k_m/n_m)^{2/d}\to0\).
\end{theorem}

\subsection{When the representation is estimated}

An embedding may be identifiable only up to a common Lorentz isometry. This is
not a problem for LHWR. Statistical latent-space modelling and inference
\citep{papamichalis2021,lixuzhu2023} concern the first-stage problem itself;
the results below concern its consequences for the second-stage regression.

\begin{proposition}[Lorentz-isometry invariance]
\label{prop:lorentz-isometry-invariance}
A simultaneous time-orientation-preserving Lorentz transformation
\(L\), with \(L^\top\bm J L=\bm J\) and
\(\bm J=\diag(-1,\bm I_d)\), of every
training coordinate and the target leaves all distances, weights,
coefficients, fitted values, and predictions unchanged, with response,
covariates, kernel, and bandwidth held fixed. A deterministic bandwidth search
based on these quantities is invariant as well when its tie rule is unchanged.
\end{proposition}

For distinct \(z,q\), let \(u_{zq}\) and \(u_{qz}\) be the unit initial
tangent vectors along their connecting geodesic; explicitly,
\(u_{zq}=\log_z(q)/d_\kappa(z,q)\) and
\(u_{qz}=\log_q(z)/d_\kappa(z,q)\).

\begin{proposition}[First variation of distance]
\label{prop:first-variation-distance}
For distinct \(z,q\in\mathbb H_\kappa^d\) and tangent perturbations
\(a\in T_z\mathbb H_\kappa^d\) and
\(b\in T_q\mathbb H_\kappa^d\), write \(z_t=\exp_z(ta)\) and
\(q_t=\exp_q(tb)\). Then
\begin{equation}
  \left.\frac{d}{dt}d_\kappa(z_t,q_t)\right|_{t=0}
  =-g_z(a,u_{zq})-g_q(b,u_{qz}).
  \label{eq:first-variation-distance}
\end{equation}
\end{proposition}

Write \(d_i=d_\kappa(z,Z_i)\),
\(\widehat d_i=d_\kappa(\widehat z,\widehat Z_i)\), and
\(\rho_n(z)=\max_i|\widehat d_i-d_i|/h\).
Equation~\eqref{eq:first-variation-distance} measures the directional change
in distance when both endpoints move; it is not asserted at \(z=q\), where
the unit directions are undefined. Its exponential maps are defined in
Equation~\eqref{eq:exp-map-explicit}. Hats here denote estimated coordinates after any common
isometry alignment, and \(\rho_n\) measures distance error relative to the
smoothing radius.

\begin{assumption}[Generated-distance regularity]
\label{ass:generated-representation}
In addition to Assumption~\ref{ass:pointwise-asymptotics}, \(K\), extended by
zero beyond 1 on \([0,\infty)\), is continuously differentiable with
Lipschitz derivative. The estimated representation uses external information
or cross-fitting so that
an observation's response error is not used to construct its own distance.
Suppose \(\rho_n(z)=o_p(1)\). The minimum eigenvalues of
\(h^{-d}\widehat{\bm\Gamma}_{n,h}^{O}\) and
\(h^{-d}\widehat{\bm\Gamma}_{n,h}^{E}\), with the \(n^{-1}\)
normalization of Equation~\eqref{eq:local-population-moments}, are bounded away
from zero with probability tending to one. With
\(r_i^O(z)=Y_i-\bx_i^\top\widehat\bbeta_h^O(z)\), assume also
\begin{equation*}
  \frac1n\sum_{i:d_i\leq2h}\|\bx_i\|\,|r_i^O(z)|=O_p(h^d).
\end{equation*}
The bisquare kernel satisfies the stated differentiability requirement.
\end{assumption}

\begin{theorem}[Generated-distance expansion]
\label{thm:generated-distance-expansion}
Under Assumption~\ref{ass:generated-representation}, let
\(\widehat\bbeta_h^{E}\) and \(\widehat\bbeta_h^{O}\) denote estimates
using estimated and oracle distances. Then
\begin{equation}
  \widehat\bbeta_h^{E}(z)-\widehat\bbeta_h^{O}(z)
  =\bm R_{n,h}(z)+O_p\{\rho_n(z)^2\},
  \label{eq:generated-distance-expansion}
\end{equation}
where the linearized first-stage distance contribution is
\begin{equation}
 \bm R_{n,h}(z)=
 \{\widehat{\bm\Gamma}_{n,h}^{O}(z)\}^{-1}
 \frac1n\sum_i K'(d_i/h)\frac{\widehat d_i-d_i}{h}\,
                    \bx_i r_i^O(z).
 \label{eq:generated-distance-linear-term}
\end{equation}
Consequently the difference is \(O_p\{\rho_n(z)\}\).
\end{theorem}
Equation~\eqref{eq:generated-distance-linear-term} is a \(p\)-vector:
distance perturbations change weights, and the oracle inverse Gram matrix
translates the perturbed residual score into coefficients.
Equation~\eqref{eq:generated-distance-expansion} controls the remaining
quadratic error; the expansion by itself does not assert a first-stage CLT.

\begin{theorem}[Joint response--representation limit]
\label{thm:joint-representation-limit}
Suppose Assumption~\ref{ass:generated-representation} holds,
\(nh^{d+8}\to0\), and \(\sqrt{nh^d}\rho_n(z)^2\xrightarrow{p}0\).
If jointly
\begin{equation}
  \begin{pmatrix}
  \sqrt{nh^d}\{\widehat\bbeta_h^O(z)-\bbeta_h(z)\}\\
  \sqrt{nh^d}\bm R_{n,h}(z)
  \end{pmatrix}
  \Rightarrow
  \begin{pmatrix}\bm G_Y(z)\\\bm G_R(z)\end{pmatrix},
  \label{eq:joint-score-limit}
\end{equation}
then
\begin{equation}
  \sqrt{nh^d}\{\widehat\bbeta_h^E(z)-\bbeta(z)-h^2\bm B(z)\}
  \Rightarrow \bm G_Y(z)+\bm G_R(z).
  \label{eq:joint-representation-limit}
\end{equation}
When the scaled first-stage term is negligible, the oracle limiting law is
recovered. Otherwise its contribution must be retained in the joint law.
\end{theorem}
Equation~\eqref{eq:joint-score-limit} is an additional joint convergence
assumption; cross-fitting alone does not establish it or independence of its
components. If the joint limit is Gaussian, with
\(\E\bm G_R=\bm b_R\), \(\Var(\bm G_R)=\bm V_R\), and
\(\Cov(\bm G_Y,\bm G_R)=\bm C_{YR}\), then
Equation~\eqref{eq:joint-representation-limit} has the law
\begin{equation}
 N_p\{\bm b_R,\ \bm V+\bm V_R+\bm C_{YR}+\bm C_{YR}^\top\}.
 \label{eq:joint-gaussian-covariance}
\end{equation}
All covariance blocks in Equation~\eqref{eq:joint-gaussian-covariance} are
\(p\times p\). A sufficient condition for oracle equivalence is
\(\sqrt{nh^d}\rho_n(z)\to_p0\); a nonnegligible first-stage term must
instead be estimated or propagated under a justified joint law.

\begin{corollary}[Generated target]
\label{cor:generated-target-location}
Suppose the training representation is fixed and
\(\widehat z=\exp_z(v_n)\), with \(v_n\in T_z\mathbb H_\kappa^d\).
Let \([v_n]_z\in\R^d\) be its coordinates in an orthonormal frame, and
let the \(p\times d\) matrix \(\bm D\bbeta(z)\) represent the
differential of \(\bbeta\) in that frame. In addition to
Assumption~\ref{ass:pointwise-asymptotics}, assume
\begin{equation*}
  \|v_n\|=o_p(h),\qquad
  \sqrt{nh^d}\|v_n\|^2\xrightarrow{p}0,\qquad nh^{d+8}\to0.
\end{equation*}
For a deterministic \(r_n=o(h)\) with
\(\Pr(\|v_n\|\leq r_n)\to1\), require stochastic equicontinuity of the
centred local estimation error:
\begin{equation*}
\begin{aligned}
 \sup_{\|v\|\leq r_n}\sqrt{nh^d}\big\|
 &[\widehat\bbeta_h\{\exp_z(v)\}-\bbeta_h\{\exp_z(v)\}]\\
 &-[\widehat\bbeta_h(z)-\bbeta_h(z)]\big\|\xrightarrow{p}0.
\end{aligned}
\end{equation*}
Assume the joint convergence
\begin{equation*}
  \begin{pmatrix}
  \sqrt{nh^d}\{\widehat\bbeta_h(z)-\bbeta_h(z)\}\\
  \sqrt{nh^d}[v_n]_z
  \end{pmatrix}
  \Rightarrow
  \begin{pmatrix}\bm G_Y(z)\\ U(z)\end{pmatrix}.
\end{equation*}
Then
\begin{equation}
  \sqrt{nh^d}\{\widehat\bbeta_h(\widehat z)-\bbeta(z)-h^2\bm B(z)\}
  \Rightarrow \bm G_Y(z)+\bm D\bbeta(z)U(z).
  \label{eq:generated-target-limit}
\end{equation}
If \(U(z)\sim N_d\{\bm0,\bm\Sigma_Z(z)\}\) independently of
\(\bm G_Y(z)\), the additional covariance is
\begin{equation}
  \bm D\bbeta(z)\bm\Sigma_Z(z)\bm D\bbeta(z)^\top.
  \label{eq:generated-target-variance}
\end{equation}
\end{corollary}
Equation~\eqref{eq:generated-target-limit} propagates target displacement
through the coefficient differential. Equation~\eqref{eq:generated-target-variance}
adds its covariance only under the stated independence; otherwise the two
cross-covariance terms must also be included.

\begin{corollary}[Adaptive representation stability]
\label{cor:representation-stability-knn}
Assume the model, kernel, and local regularity in
Corollary~\ref{cor:adaptive-knn}, with fixed \(p\), and the no-own-response
construction condition of Assumption~\ref{ass:generated-representation}.
For \(a\in\{O,E\}\), suppose the adaptive Gram matrices
\(k^{-1}\sum_i w_i^a\bx_i\bx_i^\top\) have minimum eigenvalues bounded
away from zero with probability tending to one. The two distance vectors use
the same index set. For the implementation's positive-distance ordering,
assume these target-to-training distances remain strictly positive under both
representations with probability tending to one. On the union
\(\mathcal N_n\) of the oracle and estimated contributing neighbourhoods,
require
\(k^{-1}\sum_{i\in\mathcal N_n}
(\|\bx_i\|^2+\|\bx_i\||r_i^O|)=O_p(1)\), where
\(r_i^O=Y_i-\bx_i^\top\widehat\bbeta_k^O(z)\).
If \(\delta_n=\max_i|\widehat d_i-d_i|
=o_p\{h_{n,k}^{(0)}(z)\}\), then
\begin{equation}
  \|\widehat\bbeta_k^E(z)-\widehat\bbeta_k^O(z)\|
  =O_p\{\delta_n/h_{n,k}^{(0)}(z)\}=o_p(1).
  \label{eq:representation-stability-knn-rate}
\end{equation}
Whenever the oracle adaptive estimator admits a \(\sqrt{k}\)-scaled limiting
law, the stronger condition
\(\sqrt{k}\delta_n/h_{n,k}^{(0)}(z)\xrightarrow{p}0\) transfers that law to the
estimated-representation estimator.
\end{corollary}
Equation~\eqref{eq:representation-stability-knn-rate} is a stability and
oracle-equivalence bound, not an unconditional adaptive joint-limit theorem.

The practical conclusion is straightforward. A fixed, substantively chosen
feature map can be analyzed conditionally. A learned representation should be
re-estimated within validation folds, and its uncertainty should be propagated
when it is not negligible relative to the local smoothing error.

\section{Simulation Lessons}
\label{sec:simulation}

The simulation study separates locality from geometry. In each replication,
\(n=500\) locations are sampled uniformly by hyperbolic area from domains with
\(R_{\max}\in\{4,6,8,10\}\). Coefficient surfaces are generated only from
Lorentz-hyperbolic distances. A localized regime uses true neighbourhood size
\(k_0=20\); a smooth regime uses \(k_0=80\). Ordinary least squares (OLS),
LHWR, Poincar\'e-E, and Tangent-E use identical responses, covariates, kernels,
and candidate bandwidths. The local methods differ only in distance. For these
matched designs, the candidate grid is \(k=10,20,\ldots,250\), the heterogeneity
multiplier is 2, and Gaussian noise has variance equal to one sixth of the
realized signal variance. The coefficient fields are sums of five intrinsic
bisquare bumps; their fixed target locations, amplitudes, and seed schedule
are specified in the archived simulation code. They are controlled
geometry-aligned designs, not an empirical assertion that real coefficient
surfaces must have this structure.

To distinguish selection from evaluation, let
\(\eta_i=\bx_i^\top\bbeta(Z_i)\) be the noise-free simulated response,
and let \(\widehat\eta_i=\bx_i^\top\widehat\bbeta(Z_i)\) use the fit
after bandwidth selection. We measure signal and coefficient recovery by
\begin{equation}
 \mathrm{RMSE}_{\eta}
 =\left\{n^{-1}\sum_i(\widehat\eta_i-\eta_i)^2\right\}^{1/2},\qquad
 \mathrm{RMSE}_{\beta}
 =\left\{n^{-1}\sum_i\|\widehat\bbeta(Z_i)-\bbeta(Z_i)\|_2^2\right\}^{1/2}.
 \label{eq:simulation-recovery-metrics}
\end{equation}
Equation~\eqref{eq:simulation-recovery-metrics} evaluates recovery at the
sampled locations using known simulation truth. The coefficient norm sums
over all \(p\) coefficients, including the intercept, without division by
\(p\). CV-selected identifies how the bandwidth was chosen; these
recovery errors are not held-out prediction errors.

\input{tables/tas_simulation_core.tex}

Table~\ref{tab:simulation-core-results} gives the central practical result.
For localized heterogeneity, LHWR has the smallest CV-selected coefficient
error, and its advantage remains in the oracle comparison. For the smooth
surface, Tangent-E is best and Poincar\'e-E is close. Negative curvature is
therefore not a generic performance bonus. It helps when the response surface
is locally aligned with the hyperbolic neighbourhoods; a single tangent chart
can be sufficient when variation is broad and smooth.

The rate experiment in Figure~\ref{fig:rate-validation} checks the compact
theory above. Fixed and adaptive bias slopes are \(-0.336\) and \(-0.342\),
close to the predicted \(-1/3\). At \(n=2000\), the scaled intercept variance
is 9.68 versus the analytic value 9.94. At that sample size, bias-centered standardized errors have
near-zero means, standard deviations near one, and coverage between 0.936 and
0.952.

\begin{figure}[!t]
  \centering
  \includegraphics[width=0.88\textwidth]{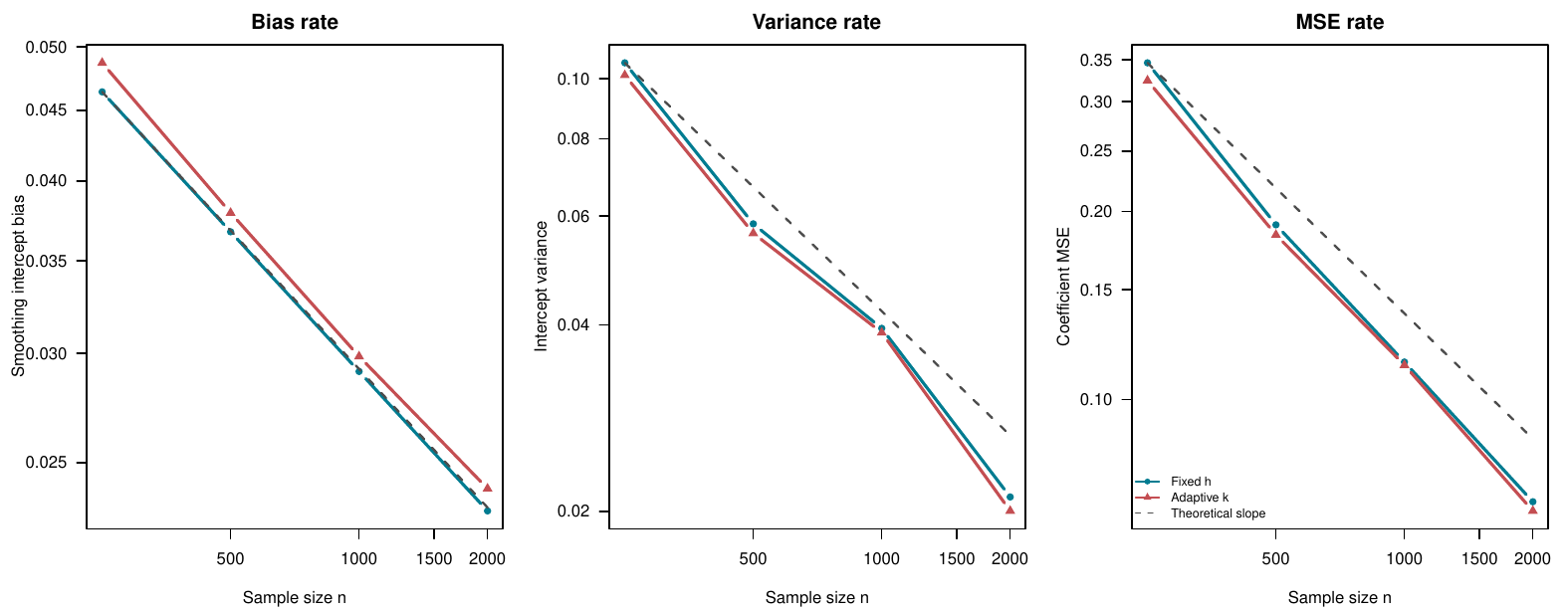}
  \caption{Monte Carlo validation of fixed- and adaptive-bandwidth rates.
  Dashed reference lines show the theoretical slopes.}
  \label{fig:rate-validation}
\end{figure}

When coordinates are estimated, conditional intervals can be too narrow.
Table~\ref{tab:generated-representation-validation} shows that adding the
target-location covariance substantially improves coverage in both first-order
and representation-dominant regimes. The adjusted slope-coefficient coverage
is 0.928--0.938 at \(n=16000\), below the nominal 0.95 target; some
undercoverage also remains in the fixed-representation benchmark. With 1000
replications, the Monte Carlo standard error of a coverage estimate in this
range is approximately 0.008. Thus the experiment demonstrates the importance
of representation uncertainty, not exact finite-sample calibration.

This is a controlled generated-target experiment: training locations remain
fixed within each replication and an independent standard bivariate Gaussian
vector \(U\), expressed in an orthonormal frame of \(T_z\mathbb H^2\),
perturbs the target through \(\widehat z=\exp_z(a_nU)\).
Here \(h=1.7n^{-1/6}\), and the four perturbation scales are \(a_n=0\),
\(n^{-1/2}\), \((nh^2)^{-1/2}\), and \(n^{-1/4}\), respectively.
The intervals use the known second-order bias, coefficient differential, and
leading response variance. They illustrate an oracle uncertainty calculation,
rather than a feasible procedure that estimates all of these quantities.
In the dominant regime, \(\sqrt{nh^2}a_n\) diverges; this is a stress test
outside the finite joint-limit premise of
Corollary~\ref{cor:generated-target-location}, not a direct verification of
that corollary. The experiment does not re-estimate a training-network embedding.

\input{tables/tas_generated_representation.tex}

\section{Economic Similarity as Locality: A WDI Illustration}
\label{sec:wdi}

The application uses 2024 World Bank World Development Indicators (WDI)
\citep{worldbank2025wdi}. After complete-case filtering, \(n=141\) countries
remain. The associational model is
\begin{equation*}
  y_i=\beta_0(z_i)+\beta_1(z_i)\,\text{unemployment}_i+
  \beta_2(z_i)\,\text{internet users}_i+\varepsilon_i,
\end{equation*}
where \(y_i\) is annual GDP growth and both covariates are standardized.

The representation uses standardized log GDP per capita and standardized
trade openness, denoted \(\psi_{i1}\) and \(\psi_{i2}\), in
Equation~\eqref{eq:generic-feature-map-Z}. These quantities are not additional
regressors. They decide which countries contribute most to each local fit.
Consequently, countries that are geographically distant can borrow strength
when their income and trade-integration profiles are similar. This answers a
different question from GWR: how do the unemployment and
internet-use associations vary across economic peer groups?

Figure~\ref{fig:wdi-geometry} shows the same representation in the Lorentz
model and Poincar\'e disk. The bounded disk is easier to read, but LHWR uses
the Lorentz-model coordinates in the left panel to calculate distances.

\begin{figure}[!t]
  \centering
  \includegraphics[width=0.94\textwidth]{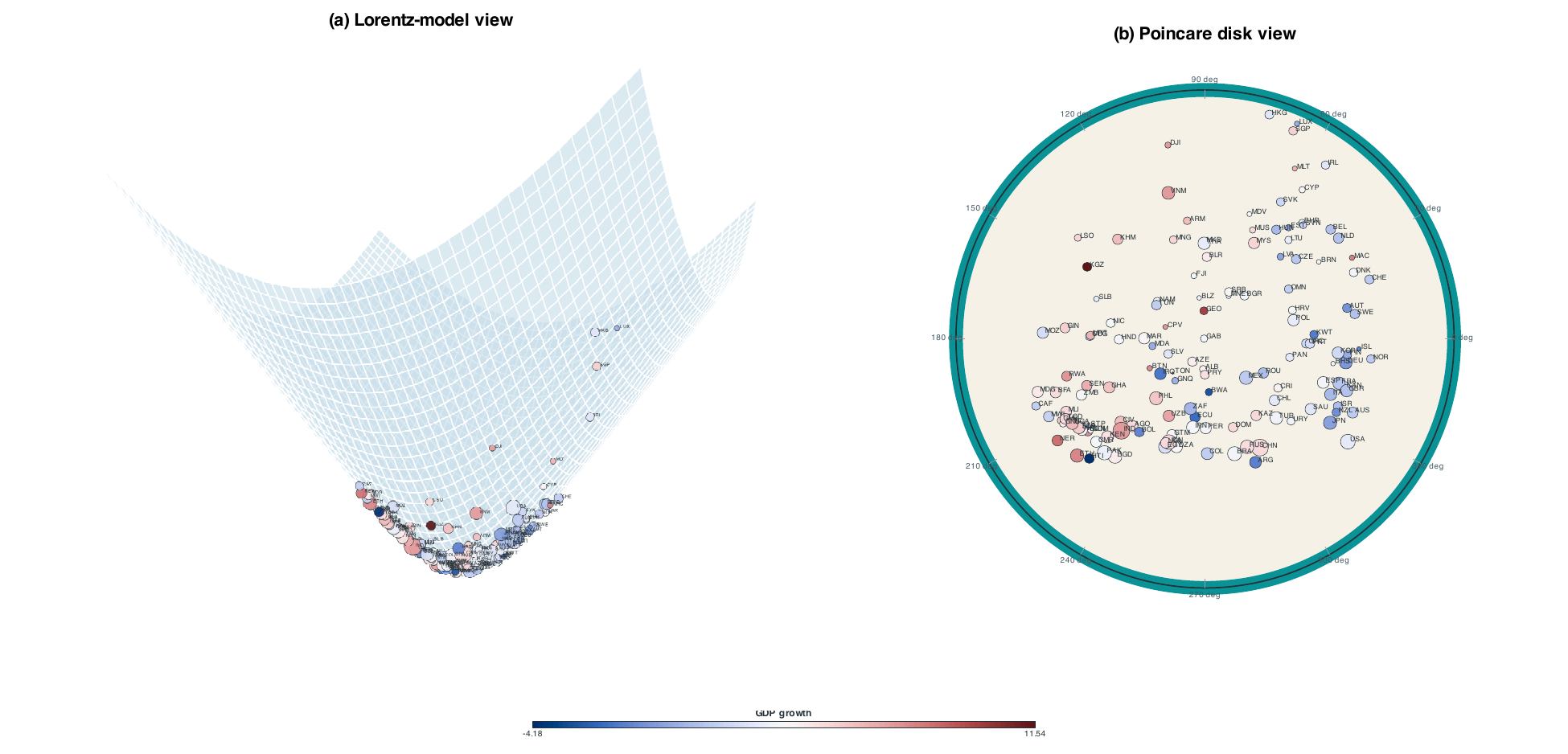}
  \caption{The 141-country WDI sample in two isometric views of the same
  economic representation. Fill color gives annual GDP growth, point size
  reflects population, and labels identify countries.}
  \label{fig:wdi-geometry}
\end{figure}

Prediction is assessed by leave-one-out
cross-validation conditional on the fixed economic representation: each
country is omitted once, and every local method reselects its bandwidth by
\AICc using the remaining responses. The feature standardization and supplied
coordinates are held fixed; the Tangent-E reference is recomputed from each
training subset. This evaluates response prediction for the represented
countries, not an end-to-end representation-learning procedure for new
countries. OLS tests the value of locality; GWR uses great-circle
distance between capital-city coordinates; Poincar\'e-E and Tangent-E use
Euclidean approximations to the same economic representation.
All four local methods use the same bisquare kernel and candidate grid
\(k=20,30,\ldots,120\), with the radius set by the \(k\)th strictly positive
training distance. Zero-distance observations retain weight one, and the
bisquare boundary has weight zero; \(k\) therefore indexes the radius rather
than counting positive weights. GWR retains great-circle distances computed
by \texttt{GWmodel}, but uses the common weighted least-squares and \AICc
implementation. This matches the tuning conventions while allowing each
method to select a different bandwidth.

For the common set of held-out responses, predictive goodness of fit is
\begin{equation}
 R^2_{\mathrm{pred}}=1-
 \frac{\sum_i(y_i-\widehat y_i^{(-i)})^2}
      {\sum_i(y_i-\bar y)^2},\qquad
 \bar y=n^{-1}\sum_i y_i.
 \label{eq:wdi-predictive-r2}
\end{equation}
In Equation~\eqref{eq:wdi-predictive-r2}, \(\widehat y_i^{(-i)}\) is the
prediction with response \(i\) omitted from fitting and bandwidth selection.
The denominator is a common descriptive normalization, not the loss of a
separately cross-validated intercept-only learner; negative values are possible.

\input{tables/tas_wdi_loocv.tex}

Table~\ref{tab:wdi-feature-map-predictive} shows that LHWR has the smallest
held-out \RMSE (2.216) and \MAE (1.678), and the largest predictive
\(R^2\) (0.223). GWR has \RMSE 2.519 and \MAE 1.879 under the matched
tuning protocol. Its excess \RMSE over LHWR is 0.302, with a paired 95\%
interval of [0.040, 0.568]. Paired \RMSE intervals favor LHWR over OLS,
GWR, and Poincar\'e-E. The interval for Tangent-E includes zero, so this
comparison does not establish a clear advantage over Tangent-E.
These percentile intervals
resample 141 paired, already computed prediction errors, with 2000 resamples;
the models are not refitted. They are descriptive
comparisons conditional on these predictions, not fully calibrated inference
for the entire learning procedure. Overlapping training folds and dependence
between countries are not reproduced by this bootstrap. The main empirical improvement
therefore comes from replacing global or geographic pooling with economic
peer groups. The exact Lorentz distance adds only a modest predictive gain for
this smooth two-feature representation.

Figure~\ref{fig:wdi-coefficients} maps the local coefficient estimates and
their pointwise significance classifications. Unemployment has raw pointwise
evidence at 48 locations, 23 of which remain after coefficient-wise
Benjamini--Hochberg adjustment. All 23 estimates are negative and occur mainly
in lower-income and lower-middle-income economic neighbourhoods, including
Ethiopia, Niger, Haiti, Pakistan, and Cameroon. Conditional on internet use,
these estimates describe a more negative local association between
unemployment and GDP growth in those peer groups. Internet use has raw
pointwise evidence at 15 locations, but none remains after adjustment. It
should therefore be interpreted as exploratory rather than as a stable set of
country-specific findings.

\begin{figure}[!t]
  \centering
  \includegraphics[width=0.98\textwidth]{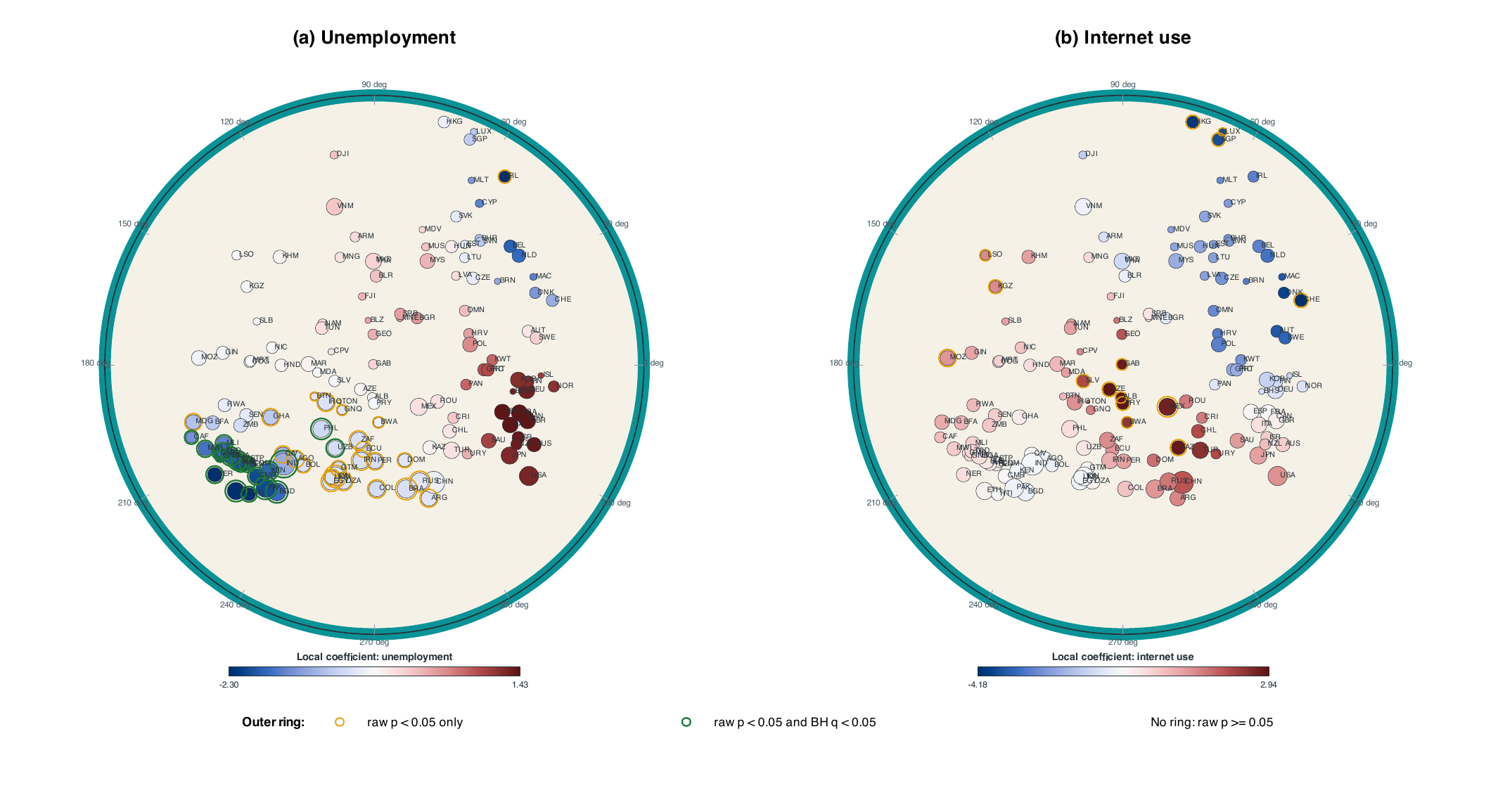}
  \caption{Local coefficients for unemployment and internet use over the WDI
  economic representation. Fill color gives the estimate. Orange rings mark
  raw pointwise \(p<0.05\) only; dark green rings mark both raw significance
  and Benjamini--Hochberg-adjusted \(q<0.05\).}
  \label{fig:wdi-coefficients}
\end{figure}

The local weighted designs are well conditioned: the median and maximum
condition numbers are 1.31 and 2.02, and every local variance inflation factor
is below 1.58. Using the same \(k=40\) diagnostic weights, OLS residuals have
Lorentz-hyperbolic Moran's \(I=0.0717\) with permutation \(p=0.020\), whereas
LHWR residuals have \(I=-0.0383\) with \(p=0.140\). This descriptive contrast
suggests that the positive economic-neighbourhood pattern remaining after
global fitting is not retained after local calibration.

Using the same zero-diagonal, row-standardised bisquare weights and 99
full-vector residual permutations, the corresponding LISA
diagnostic flags only Bolivia as a low--high residual
outlier (local \(I_i^{(H)}=-0.6193\), raw two-sided permutation \(p=0.020\));
the other 140 countries are not flagged at the 5\% threshold. Bolivia has a
below-mean residual but a positive weighted lag of centred residuals among
its economic peers. This is an isolated discordant pattern, not a high--high
or low--low residual cluster. The single unadjusted flag among 141 local
diagnostics is exploratory, not a multiple-testing-adjusted discovery or
evidence that all residual dependence has been eliminated.

\section{Choosing Among LHWR and Its Approximations}
\label{sec:guidance}

The simulations and application support a conditional recommendation rather
than a winner-takes-all rule.

\begin{itemize}
  \item Use LHWR when the representation has a defensible negative-curvature
  interpretation and sharply local peer groups are scientifically plausible.
  \item Use Tangent-E as a sensitivity analysis and as a simpler alternative
  when observations occupy a moderate region around one reference point or
  the coefficient surface appears smooth.
  \item Use Poincar\'e-E only as a projection-based ablation. Ordinary
  Euclidean distance in the Poincar\'e disk is not hyperbolic distance.
  \item Retain a global model to show whether locality matters at all, and use
  GWR only when physical distance is a scientifically relevant
  competing notion of locality.
\end{itemize}

The most consequential modeling decision is often the representation, not the
choice among closely related metrics. A feature map based on economic
development defines different peer groups from one based on institutions,
financial structure, or trade networks. Such alternatives should be chosen
from subject-matter reasoning and assessed by held-out prediction and
stability, not selected only because they produce a visually appealing map.

\section{Discussion}
\label{sec:discussion}

LHWR turns a hyperbolic representation into an interpretable rule for
borrowing information: the geometry defines local peers, and regression
coefficients describe how covariate associations vary among them. This
connects representation learning with the varying-coefficient framework
\citep{hastietibshirani1993}. The theory, prediction rule, bandwidth selection,
and diagnostics make that locality rule statistically explicit and
empirically assessable.

A hyperbolic representation is particularly well motivated by hierarchical,
tree-like, or core--periphery organization.
Equation~\eqref{eq:exact-ball-volume} defines its neighbourhood volume;
Equation~\eqref{eq:ball-volume-growth} establishes exponential growth at large
geodesic radii in dimension at least two. Unlike polynomial Euclidean volume
growth, this capacity accommodates the proliferation of branches across
hierarchical levels \citep{krioukov2010,nickel2017,nickel2018}. It concerns
representational capacity, not temporal response growth or guaranteed
prediction gains. This large-radius motivation must also be distinguished
from the small-neighbourhood limit underlying the smoothing theory.

For regression, nearby represented observations should plausibly have similar
coefficient relationships; hubs or hyperboloid coordinates alone do not
establish this. When observations lie near a common centre relative to the
curvature radius, Equation~\eqref{eq:tangent-distortion-bound} controls
tangent-distance distortion. Tangent-E may then suffice, although regression
performance also depends on the coefficient surface and selected
neighbourhoods. Geometry comparisons therefore complement, rather than
replace, substantive justification.

For fixed representations, the pointwise results establish consistency,
bias--variance expansions, and asymptotic normality under the stated
conditions. Familiar local smoothing rates are retained. In
Proposition~\ref{prop:curvature-contribution}, the explicit scalar-curvature
term cancels from leading local-constant coefficient bias, but geometry
remains in intrinsic derivatives, neighbourhood volume, and higher-order
variance. The adaptive results connect neighbour counts with metric radii
and estimation error.

For estimated representations, treating learned coordinates as fixed can
omit variation of the same order as regression noise. In the first-order regime of
Table~\ref{tab:generated-representation-validation}, adding the target-location
covariance raises coverage from 0.885--0.891 to 0.928--0.935. This controlled
experiment uses oracle bias and variance quantities and still falls below
nominal coverage. It demonstrates the importance of representation
uncertainty, not a fully calibrated procedure for arbitrary embeddings.
The joint-limit results require a suitable first-stage limit; they do not
establish one for every embedding algorithm.

Table~\ref{tab:simulation-core-results} shows the smallest coefficient-recovery
error for LHWR under localized heterogeneity, whereas Tangent-E performs best
for the broad, smooth surface. These controlled, intrinsically generated
coefficient fields separate locality from distance approximation. They
support a conditional choice of geometry, not a universal ranking.

The WDI illustration gives that distinction an empirical interpretation.
For the common 141-country sample,
Table~\ref{tab:wdi-feature-map-predictive} reports held-out \RMSE values of
2.216 for LHWR, 2.446 for OLS, and 2.519 for GWR under matched tuning
conventions. The results are consistent with a benefit from pooling countries
by economic similarity instead of imposing one global relationship or using
physical proximity. The additional gain from exact Lorentz distance is
modest: Tangent-E has \RMSE 2.236, and its paired difference interval includes
zero. The two-feature WDI map is therefore an illustration of
representation-defined locality, not evidence that the observed countries
form a hierarchically growing network or that negative curvature is necessary
for this application.

The maps complement aggregate prediction scores:
Figure~\ref{fig:wdi-coefficients} identifies 23 locations with
negative unemployment coefficients after coefficient-wise multiplicity
adjustment, while no internet-use coefficient remains significant after that
adjustment. These conditional, exploratory findings identify associations for
further investigation, not causal policy effects. Country dependence and
overlapping training folds limit the paired-error intervals; residual
diagnostics do not establish the independence assumed in the pointwise
theory. Fixed-map validation also does not assess end-to-end representation
learning for new countries.

\begingroup
\clubpenalty=10000
Computational scalability remains a separate question. Fast,
high-performance, and GPU implementations have been developed for GWR
\citep{li2019fast,lu2022hpgwr,wang2020cuda,cuda2025gwr}. LHWL supplies dense
and exact nearest-neighbour implementations, but the present study does not
benchmark their large-scale performance against those systems. Further work
should develop feasible uncertainty propagation for learned representations
under dependence, alongside scalable computation.
\par
\endgroup

\section{Conclusion}
\label{sec:conclusion}

LHWR makes hyperbolic representations useful for interpretable local
regression. Geometry determines whose information enters a fit, while the
coefficient field describes the resulting associations. The fixed- and
estimated-representation theory states when estimation is stable and when
representation uncertainty must enter inference.

The evidence shows both the value and the limits of this approach. LHWR is
advantageous for localized simulated heterogeneity; tangent approximations
can be competitive or better for smoother variation. In the WDI application,
economic-peer regression improves on global and geographic pooling, and the
maps locate evidence for particular covariate associations.

The practical gain is a scientifically motivated definition of locality,
coupled with explicit estimation and uncertainty assessment. Hyperbolic
structure is useful when it preserves relationships relevant to the
regression question. Its value should be demonstrated through comparisons
and diagnostics, not inferred from geometry alone.

\paragraph*{Data availability statement.}
The WDI data are publicly available from the World Bank. The processed
analysis data contain only redistributable indicators and country identifiers.
The computational reproducibility archive contains the processed sample,
download and preprocessing scripts, and the recorded outputs. The archive is
available from the corresponding author on request. A public repository
identifier will be provided when the archive is publicly deposited.

\paragraph*{Code availability.}
The LHWR method is implemented in the open-source LHWL R package, distributed
under the GNU General Public License (version 3 or later). The package
repository is \url{https://github.com/byuzbasi/LHWL}. The computational
reproducibility archive includes the package source, scripts for reproducing
the analyses, tables, and figures, and software-version records.

\paragraph*{Conflict of interest.}
The authors declare no competing interests.

\paragraph*{Declaration of generative AI assistance.}
During manuscript preparation, ChatGPT (OpenAI) was used for English-language
editing and stylistic refinement. The authors reviewed all content and take
full responsibility for the accuracy and scientific integrity of the manuscript.

\FloatBarrier
\appendix
\numberwithin{equation}{section}
\counterwithin{table}{section}

\section{Geometric Definitions and Conventions}
\label{app:geometric-definitions}

We use the upper-sheet model, distance, and tangent metric in Equations
\eqref{eq:lorentz-distance-general} and
\eqref{eq:tangent-space-metric}. The identities below fix conventions
used in the proofs \citep{helgason1978,chavel2006,robbin2024diffgeo}.

\begin{proof}[Proof of Proposition~\ref{prop:geodesic-metric}]
A time-orientation-preserving Lorentz transformation sends the first point
to \((R_\kappa,0,\ldots,0)\) without changing inner products. The upper-sheet
constraint makes the time component of the transformed second point at least
\(R_\kappa\), with equality only at that origin. Thus
\(-\innerL{z}{w}/R_\kappa^2\geq1\), proving nonnegativity and separation;
symmetry follows from the inner product. A radial unit-speed geodesic has
coordinates \((R_\kappa\cosh(t/R_\kappa),
R_\kappa\sinh(t/R_\kappa)\omega)\) for a unit spatial direction
\(\omega\). Its length from the origin is \(t\), which gives the stated
\(\acosh\) distance. Hyperbolic space is complete and simply connected
with nonpositive curvature, so this geodesic minimizes length. The infimum
of curve lengths satisfies the triangle inequality by concatenation.
\end{proof}

We verify the inverse relation between Equations~\eqref{eq:tangent-log-map}
and~\eqref{eq:exp-map-explicit}.
For \(r=d_\kappa(\mu,z)>0\), the vector
\(z+R_\kappa^{-2}\innerL{\mu}{z}\mu\) is tangent at \(\mu\) and has
norm \(R_\kappa\sinh(r/R_\kappa)\). Multiplying by the factor in
Equation~\eqref{eq:tangent-log-map} therefore gives a vector of norm
\(r\), and substitution in Equation~\eqref{eq:exp-map-explicit} returns
\(z\). At equality the continuous definitions give
\(\log_\mu(\mu)=0\) and \(\exp_\mu(0)=\mu\). There is no cut locus:
the exponential and logarithmic maps are mutually inverse on the full
tangent space and upper sheet. The unit vector \(\log_\mu(z)/r\), however,
is not defined at \(z=\mu\); this explains the distinct-endpoint condition
in the distance-variation proposition.

The default reference in Equation~\eqref{eq:tangent-reference} is
well defined for every finite nonempty upper-sheet sample, since
\begin{equation}
 -\innerL{\bar z_{\rm tr}}{\bar z_{\rm tr}}
 =n_{\rm tr}^{-2}\sum_{i,j\in\mathcal I_{\rm tr}}
       \{-\innerL{z_i}{z_j}\}\ \geq\ R_\kappa^2>0.
 \label{eq:reference-existence}
\end{equation}
Equation~\eqref{eq:reference-existence} follows from the metric proposition;
the positive time component of the mean selects the upper sheet. Normalizing
this timelike vector commutes with a common Lorentz isometry. Tangent
coordinates in two orthonormal bases differ by an orthogonal matrix, so
Equation~\eqref{eq:tangent-distance} is basis independent. These facts
do not identify the extrinsic reference with a squared-geodesic Fr\'echet
mean or remove the effect of choosing a different reference.

In coordinates \(x^1,\ldots,x^d\), let \(g_{ab}\) be the metric matrix,
\(g^{ab}\) its inverse, and \(|g|=\det(g_{ab})\). For a twice
differentiable scalar field \(F\), the convention used in the bias formula is
\begin{equation}
 \Delta_\kappa F
 =\frac{1}{\sqrt{|g|}}\sum_{a,b=1}^d
   \partial_a\{\sqrt{|g|}\,g^{ab}\partial_b F\}.
 \label{eq:laplace-beltrami-definition}
\end{equation}
Equation~\eqref{eq:laplace-beltrami-definition} is the trace of the covariant
Hessian. In a normal orthonormal frame at the target it becomes
\(\sum_a\partial_a^2 F\); with suitable zero boundary terms,
\(\int F\Delta_\kappa F\,dV_\kappa
=-\int\|\operatorname{grad}_gF\|^2dV_\kappa\).
This fixes the sign, and all vector/matrix derivatives below are entrywise.
Contraction of constant sectional curvature gives
\(\operatorname{Ric}_\kappa=(d-1)\kappa g\) and its trace
\(\operatorname{Scal}_\kappa=d(d-1)\kappa\), as stated in
Equation~\eqref{eq:curvature-contractions}.

Let \(r=\|v\|_z\) be normal-coordinate radius and \(d\omega^2\) the unit
sphere metric. The polar metric and the Jacobian relative to tangent
Lebesgue measure \(dv\) are
\begin{equation}
 ds^2=dr^2+\{R_\kappa\sinh(r/R_\kappa)\}^2d\omega^2,
 \qquad
 J_z(v)=\left\{\frac{\sinh(r/R_\kappa)}{r/R_\kappa}\right\}^{d-1}.
 \label{eq:polar-metric-jacobian}
\end{equation}
Equation~\eqref{eq:polar-metric-jacobian} means
\(dV_\kappa\{\exp_z(v)\}=J_z(v)dv\), with the continuous value
\(J_z(0)=1\). Expanding it gives
\(J_z(v)=1-\operatorname{Ric}_\kappa(v,v)/6+O(r^4)\), the volume factor
used in Appendix~\ref{app:normal-coordinate-expansion}.

The same polar metric quantifies the effect of measuring distances in a
single tangent space rather than intrinsically on the hyperboloid.
\begin{lemma}[Tangent-space distance distortion]
\label{lem:tangent-distance-distortion}
Fix \(\kappa<0\), \(R_\kappa=(-\kappa)^{-1/2}\), and a common centre
\(\mu\in\mathbb H_\kappa^d\). Let \(\rho\geq0\) and suppose
\(d_\kappa(\mu,z_i),d_\kappa(\mu,z_j)\leq\rho\). For the exact
logarithmic map and the tangent distance in
Equation~\eqref{eq:tangent-distance},
\begin{equation}
 d_{T,\mu}(z_i,z_j)\leq d_\kappa(z_i,z_j)
 \leq C_\rho\,d_{T,\mu}(z_i,z_j),\qquad
 C_\rho=\frac{\sinh(\rho/R_\kappa)}{\rho/R_\kappa},
 \label{eq:tangent-distortion-bound}
\end{equation}
where \(C_0=1\) is the continuous value.
\end{lemma}

\begin{proof}
Write \(v_a=\log_\mu(z_a)\) for \(a\in\{i,j\}\). On the tangent space,
the Euclidean polar metric induced by \(g_\mu\) is
\(ds_E^2=dr^2+r^2d\omega^2\), where \(r=\|v\|_\mu\).
Equation~\eqref{eq:polar-metric-jacobian}, now centred at \(\mu\), gives
the pullback of the hyperbolic metric under \(\exp_\mu\).
The ratio \(\sinh x/x\), extended by one at zero, is at least one and
increases for \(x\geq0\): its derivative has numerator
\(x\cosh x-\sinh x\), whose derivative is \(x\sinh x\geq0\).
Consequently every curve has hyperbolic length at least the Euclidean length
of its logarithmic image. Taking the infimum over curves joining the two
points gives the first inequality in Equation~\eqref{eq:tangent-distortion-bound}.
For the second inequality, the straight segment from \(v_i\) to \(v_j\)
stays in the convex tangent ball \(\|v\|_\mu\leq\rho\). On this segment,
the hyperbolic polar metric is bounded above by \(C_\rho^2ds_E^2\).
Its exponential image therefore has length at most
\(C_\rho\|v_i-v_j\|_\mu\), which bounds the geodesic distance above.
The metric comparison extends continuously across \(r=0\); in dimension
one there is no angular contribution and the distances are equal.
\end{proof}

To interpret Lemma~\ref{lem:tangent-distance-distortion} near its common
centre, Taylor expansion of the multiplier gives
\begin{equation}
 C_\rho=1+\frac{\rho^2}{6R_\kappa^2}
       +O\!\left(\frac{\rho^4}{R_\kappa^4}\right),
 \qquad \rho/R_\kappa\to0.
 \label{eq:tangent-distortion-expansion}
\end{equation}
Together, Equations~\eqref{eq:tangent-distortion-bound} and
\eqref{eq:tangent-distortion-expansion} bound the relative distance
distortion for distinct points by
\begin{equation*}
 0\leq\frac{d_\kappa(z_i,z_j)}{d_{T,\mu}(z_i,z_j)}-1
 \leq\frac{\rho^2}{6R_\kappa^2}+O\!\left(\frac{\rho^4}{R_\kappa^4}\right).
\end{equation*}
Thus the approximation is controlled by distance from the common centre
relative to the curvature radius, not merely by proximity of the two points
to one another. This is a geometric bound, not a regression-risk ordering.
It concerns the exact map: the unit-curvature implementation instead sets
the tangent vector to zero when the computed \(d_{-1}(\mu,z)\leq10^{-10}\).
For distinct points in that numerical neighbourhood, the computed tangent
distance can therefore vanish; the exact multiplicative bound is not
asserted for these thresholded distances.

If \(\omega_{d-1}=dv_d\) denotes the surface area of the Euclidean unit
sphere, integration gives the exact ball volume
\begin{equation}
 \operatorname{Vol}_\kappa\{B(z,r)\}
 =\omega_{d-1}\int_0^r
       \{R_\kappa\sinh(t/R_\kappa)\}^{d-1}\,dt.
 \label{eq:exact-ball-volume}
\end{equation}
Equation~\eqref{eq:exact-ball-volume} reduces to
\(2\pi R_\kappa^2\{\cosh(r/R_\kappa)-1\}\) for \(d=2\).
For \(d\geq2\), using \(\sinh u\sim e^u/2\) yields
\begin{equation}
 \operatorname{Vol}_\kappa\{B(z,r)\}
 \sim\frac{\omega_{d-1}R_\kappa^d}{(d-1)2^{d-1}}
       e^{(d-1)r/R_\kappa},\qquad r\to\infty.
 \label{eq:ball-volume-growth}
\end{equation}
Equation~\eqref{eq:ball-volume-growth} concerns large-radius volume, whereas
the kernel expansions concern small radii. For \(d=1\) the ball length is
\(2r\), so the exponential-growth assertion is not made.

\section{Finite-Sample Local Estimation}
\label{app:finite-sample-estimation}

\begin{proof}[Proof of Proposition~\ref{prop:local-wls}]
Let \(\bm G_z=\bX^\top\bW\bX\). For any nonzero \(a\in\R^p\),
\(a^\top\bm G_z a=\|\bW^{1/2}\bX a\|_2^2>0\) by
Assumption~\ref{ass:local-design}. The gradient of the loss in
Equation~\eqref{eq:local-wls-objective} is
\(-2\bX^\top\bW(\by-\bX\theta)\), and its Hessian is
\(2\bm G_z\), a positive-definite matrix. Its unique stationary point is
therefore the unique minimizer in Equation~\eqref{eq:lhwr-estimator}.
\end{proof}

For a fixed admissible weighted design, \(\bA_i\bX=\bI_p\).
Taking expectations and covariances in \(\widehat\bbeta(Z_i)=\bA_i\by\)
proves Equation~\eqref{eq:conditional-moments}. Subtracting
\(\bbeta(Z_i)=\bA_i\bX\bbeta(Z_i)\) gives its smoothing-bias identity.
If the conditional errors are Gaussian, linear transformation also gives
\begin{equation}
 \widehat\bbeta(Z_i)\mid(\bX,\bZ)
 \sim N_p\{\bA_i\bm m,\bA_i\bm\Sigma\bA_i^\top\}.
 \label{eq:conditional-gaussian-law}
\end{equation}
Equation~\eqref{eq:conditional-gaussian-law} is centred at \(\bA_i\bm m\),
not generally at \(\bbeta(Z_i)\). Estimated variance and response-selected
weights do not give an exact pointwise Student law from this display.

A weighted Gauss--Markov interpretation requires a separate, exact working
model. On the active set \(\mathcal J_i=\{j:w_{ij}>0\}\), write the
restricted arrays as \(\by_+,\bX_+,\bW_+\). Suppose
\begin{equation}
 \E(\by_+\mid\bX,\bZ)=\bX_+\bbeta_i,\qquad
 \Var(\by_+\mid\bX,\bZ)=\sigma_i^2\bW_+^{-1}.
 \label{eq:weighted-gauss-markov-model}
\end{equation}
Equation~\eqref{eq:weighted-gauss-markov-model} assumes an exactly constant
local coefficient and inverse-weight covariance; neither is the general
varying-coefficient model. Premultiplication by \(\bW_+^{1/2}\) produces a
full-rank homoskedastic linear model. If another linear unbiased estimator
has operator \(\bA_++\bm D\), then \(\bm D\bX_+=0\); its covariance
exceeds that of \(\bA_+\by_+\) by
\(\sigma_i^2\bm D\bW_+^{-1}\bm D^\top\), a positive-semidefinite matrix.
Thus LHWR is BLUE within this specific active-set model. Replacing the exact
mean or covariance in Equation~\eqref{eq:weighted-gauss-markov-model} by an
approximation does not establish exact unbiasedness or BLUE optimality.

\section{Normal-Coordinate Kernel Expansion}
\label{app:normal-coordinate-expansion}

This appendix proves the geometric expansion used throughout the asymptotic
theory. All derivatives are evaluated at the fixed interior target \(z\), and
the calculation is applied componentwise when the integrand is vector- or
matrix-valued.

\begin{lemma}[Radial kernel moment on \(\mathbb H_\kappa^d\)]
\label{lem:radial-kernel-moment}
For a bounded, nonnegative radial kernel \(K\) supported on \([0,1]\),
and any field \(\bm F\) with four bounded covariant derivatives near the
fixed interior target \(z\), as \(h\to0\),
\begin{align*}
 &\int_{\mathbb H_\kappa^d}
 K\left\{\frac{d_\kappa(z,q)}{h}\right\}\bm F(q)\,dV_\kappa(q)\\
 &\qquad=
 h^d\left[
   \mu_0\bm F
   +h^2\mu_2
    \left\{
      \frac{1}{2}\Delta_\kappa\bm F
      -\frac{1}{6}\operatorname{Scal}_\kappa\bm F
    \right\}
   +O(h^4)
 \right].
\end{align*}
The same formula holds with \(K^2\), \(\nu_0\), and \(\nu_2\) replacing
\(K\), \(\mu_0\), and \(\mu_2\). The moments are the finite integrals
defined in Section~\ref{sec:theory}. Lipschitz continuity of \(K\) is not
needed for this deterministic integral expansion, so the indicator kernel
\(K(r)=I(r<1)\) is included. The stronger kernel conditions in
Assumption~\ref{ass:pointwise-asymptotics} remain in force for the statistical
results that invoke that assumption.
\end{lemma}

\begin{proof}
Because the target is interior and the injectivity radius of hyperbolic space
is infinite, write \(q=\exp_z(v)\) throughout the kernel support. Set
\(v=hu\). Geodesic normal coordinates give
\begin{equation*}
  d_\kappa\{z,\exp_z(hu)\}=h\|u\|
\end{equation*}
and
\begin{equation*}
  dV_\kappa\{\exp_z(hu)\}
  =h^dJ_z(hu)\,du.
\end{equation*}
For constant sectional curvature,
\begin{equation*}
  J_z(hu)
  =1-\frac{h^2}{6}\operatorname{Ric}_\kappa(u,u)
  +O(h^4\|u\|^4).
\end{equation*}
There is no cubic derivative-of-curvature term because the curvature tensor is
parallel. Covariant Taylor expansion along the radial geodesic gives
\begin{align*}
  \bm F\{\exp_z(hu)\}
  ={}&\bm F
  +h\sum_a u_a\nabla_{e_a}\bm F
  +\frac{h^2}{2}\sum_{a,b}u_au_b
    (\nabla^2\bm F)_{ab}\\
  &+\frac{h^3}{6}\sum_{a,b,c}u_au_bu_c
    (\nabla^3\bm F)_{abc}
  +O(h^4\|u\|^4),
\end{align*}
where \(e_1,\ldots,e_d\) is a normal orthonormal frame at \(z\), and
\((\nabla^j\bm F)_{a_1\cdots a_j}\) denotes the covariant derivative
tensor evaluated on those basis vectors. In particular \(\nabla^2\bm F\)
is the entrywise Hessian; this avoids identifying it with arbitrary iterated
coordinate derivatives away from the target. Radial symmetry
implies
\begin{equation*}
  \int u_aK(\|u\|)\,du=0,
  \qquad
  \int u_au_bK(\|u\|)\,du=\mu_2I(a=b),
\end{equation*}
and all third-order moments vanish. The second derivative term therefore
integrates to \((\mu_2/2)\Delta_\kappa\bm F\). Similarly,
\begin{equation*}
  \int K(\|u\|)\operatorname{Ric}_\kappa(u,u)\,du
  =\mu_2\operatorname{Scal}_\kappa.
\end{equation*}
Multiplying the Taylor and Jacobian expansions, integrating over the unit ball,
and collecting terms gives the result. Repeating the argument with \(K^2\)
gives the second statement.
\end{proof}

The kernel constants can also be checked without simulation. For the bisquare
kernel, polar integration with \(t=r^2\) gives
\begin{equation}
\begin{aligned}
 \mu_0&=\frac{8v_d}{(d+2)(d+4)},&
 \mu_2&=\frac{8v_d}{(d+2)(d+4)(d+6)},\\
 \nu_0&=\frac{384v_d}{(d+2)(d+4)(d+6)(d+8)},&
 \nu_2&=\frac{384v_d}{(d+2)(d+4)(d+6)(d+8)(d+10)}.
\end{aligned}
\label{eq:bisquare-kernel-constants}
\end{equation}
Equation~\eqref{eq:bisquare-kernel-constants} evaluates the moments used in
the main bias and variance formulas: angular symmetry gives
\(\int u_1^2K(\|u\|)du=v_d\int_0^1r^{d+1}K(r)dr\), and the remaining
integrals are beta integrals. For \(d=2\), they are
\(\mu_0=\pi/3\), \(\mu_2=\pi/24\), \(\nu_0=\pi/5\), and
\(\nu_2=\pi/60\). At \(\kappa=-1\), substitution in
Equation~\eqref{eq:curvature-variance-correction} gives the
curvature-only multiplier \(1-h^2/18\), subject to that proposition's
additional assumptions.

\section{Proofs of Consistency and Bias--Variance Results}
\label{app:consistency-proofs}

Apply Lemma~\ref{lem:radial-kernel-moment} to \(\bm Q=f\bm M\) and to
\(\bm Q\bbeta\). Equations~\eqref{eq:asymptotic-model}--
\eqref{eq:asymptotic-Q} imply
\begin{align}
  \bm\Gamma_h(z)
  ={}&h^d\left[
    \mu_0\bm Q
    +h^2\mu_2
      \left\{
        \frac{1}{2}\Delta_\kappa\bm Q
        -\frac{1}{6}\operatorname{Scal}_\kappa\bm Q
      \right\}
    +O(h^4)
  \right],
  \label{eq:appendix-Gamma-expansion}\\
  \bm g_h(z)
  ={}&h^d\left[
    \mu_0\bm Q\bbeta
    +h^2\mu_2
      \left\{
        \frac{1}{2}\Delta_\kappa(\bm Q\bbeta)
        -\frac{1}{6}\operatorname{Scal}_\kappa\bm Q\bbeta
      \right\}
    +O(h^4)
  \right].
  \label{eq:appendix-g-expansion}
\end{align}
Since \(\bm Q(z)\) is positive definite, the inverse in
equation~\eqref{eq:population-and-sample-estimator} exists for sufficiently
small \(h\). The matrix expansion
\begin{equation*}
  (\bm A+h^2\bm C)^{-1}
  =\bm A^{-1}-h^2\bm A^{-1}\bm C\bm A^{-1}+O(h^4)
\end{equation*}
applied to equations~\eqref{eq:appendix-Gamma-expansion} and
\eqref{eq:appendix-g-expansion} yields
\begin{equation*}
  \bbeta_h-\bbeta
  =
  \frac{h^2\mu_2}{2\mu_0}\bm Q^{-1}
  \left[
    \Delta_\kappa(\bm Q\bbeta)
    -(\Delta_\kappa\bm Q)\bbeta
  \right]
  +O(h^4).
\end{equation*}
The two scalar-curvature terms cancel exactly. The covariant product rule
\begin{equation*}
  \Delta_\kappa(\bm Q\bbeta)
  =
  (\Delta_\kappa\bm Q)\bbeta
  +2\sum_{a=1}^d
    (\nabla_{e_a}\bm Q)(\nabla_{e_a}\bbeta)
  +\bm Q\Delta_\kappa\bbeta
\end{equation*}
gives equations~\eqref{eq:asymptotic-bias-constant} and
\eqref{eq:asymptotic-bias-equivalent}.

Write \(K_{h,z}(q)=K\{d_\kappa(z,q)/h\}\) and
\(\widehat{\bm\Gamma}_{n,h}=n^{-1}\sum_iw_i\bx_i\bx_i^\top\).
Define the centred local score
\begin{equation*}
  \bm\psi_{i,h}(z)
  =K_{h,z}(Z_i)\bx_i
  \{Y_i-\bx_i^\top\bbeta_h(z)\}.
\end{equation*}
By the definition of \(\bbeta_h\),
\(\E\{\bm\psi_{i,h}(z)\}=\bm0\), and the sample normal equations give
\begin{equation}
  \widehat\bbeta_h(z)-\bbeta_h(z)
  =
  \widehat{\bm\Gamma}_{n,h}(z)^{-1}
  \frac{1}{n}\sum_{i=1}^n\bm\psi_{i,h}(z).
  \label{eq:appendix-score-identity}
\end{equation}
The kernel support and moment assumptions imply
\begin{equation*}
  \widehat{\bm\Gamma}_{n,h}-\bm\Gamma_h
  =O_p\{(h^d/n)^{1/2}\}.
\end{equation*}
Equation~\eqref{eq:appendix-Gamma-expansion} and \(nh^d\to\infty\) give
\begin{equation}
  h^d\widehat{\bm\Gamma}_{n,h}^{-1}
  \xrightarrow{p}\{\mu_0\bm Q(z)\}^{-1}.
  \label{eq:appendix-Gram-limit}
\end{equation}
Moreover,
\begin{equation*}
  \Var\left\{
    \frac{1}{n}\sum_{i=1}^n\bm\psi_{i,h}(z)
  \right\}=O(h^d/n).
\end{equation*}
Combining these orders with equation~\eqref{eq:appendix-score-identity} gives
\begin{equation*}
  \widehat\bbeta_h(z)-\bbeta_h(z)
  =O_p\{(nh^d)^{-1/2}\}.
\end{equation*}
Adding the population smoothing bias proves
equation~\eqref{eq:pointwise-rate}. Since \(h\to0\) and \(nh^d\to\infty\),
both terms vanish, proving Theorem~\ref{thm:pointwise-consistency}.

Continuity of \(\bbeta\) and equation~\eqref{eq:population-bias-expansion}
imply, uniformly on the kernel support,
\begin{equation*}
  \bx_i^\top\{\bbeta(Z_i)-\bbeta_h(z)\}=O_p(h\|\bx_i\|).
\end{equation*}
The contribution of this local coefficient difference to the score covariance
is \(O(h^{d+2})\). The conditional mean-zero error supplies the leading term.
The squared-kernel version of Lemma~\ref{lem:radial-kernel-moment} gives
\begin{equation}
  \E\{\bm\psi_{i,h}(z)\bm\psi_{i,h}(z)^\top\}
  =h^d\nu_0f(z)\bm\Omega(z)+o(h^d).
  \label{eq:appendix-score-variance}
\end{equation}
The central limit theorem proved in the next appendix and the additional
uniform-integrability condition in Theorem~\ref{thm:bias-variance}
justify convergence of the first two moments of
\(\bm T_{n,h}=\sqrt{nh^d}(\widehat\bbeta_h-\bbeta_h)\).
In particular, \(\E\bm T_{n,h}\to\bm0\) and
\(\E(\bm T_{n,h}\bm T_{n,h}^\top)\to\bm V(z)\). Hence
\begin{equation*}
  \Var\{\widehat\bbeta_h(z)\}
  =
  \frac{1}{nh^d}
  \frac{\nu_0}{\mu_0^2f(z)}
  \bm M(z)^{-1}\bm\Omega(z)\bm M(z)^{-1}
  +o\{(nh^d)^{-1}\}.
\end{equation*}
This proves equation~\eqref{eq:variance-expansion}.

For clarity, uniform integrability is an additional moment assumption, not
a consequence of convergence in probability of the inverse Gram matrix.
Explicitly, the condition on the completed estimator is
\begin{equation*}
 \lim_{M\to\infty}\sup_n
 \E\{\|\bm T_{n,h}\|^2 I(\|\bm T_{n,h}\|^2>M)\}=0.
\end{equation*}
A uniformly bounded \((2+\eta)\)th moment of \(\bm T_{n,h}\), for some
\(\eta>0\), is sufficient. Neither bound is automatic for a compactly
supported kernel: rare, nearly singular local designs must be controlled.
Without this extra condition, \(\bm V(z)\) remains the covariance of the
limiting normal law, but convergence of unconditional variances is not
asserted.

Under the same uniform-integrability assumption, the expectation expansion
that is actually justified is
\begin{equation}
 \E\widehat\bbeta_h(z)-\bbeta(z)
 =h^2\bm B(z)+O(h^4)+o\{(nh^d)^{-1/2}\}.
 \label{eq:expectation-bias-qualified}
\end{equation}
Equation~\eqref{eq:expectation-bias-qualified} follows from
\(\E\bm T_{n,h}\to0\); it is not an unconditional \(O(h^4)\) remainder
for the sample bias. For example, \((nh^d)^{-1/2}=O(h^2)\) makes its last
term \(o(h^2)\); without such a comparison or a stronger expectation bound,
the population and finite-sample biases must remain distinguished.

The population smoothing bias and leading covariance give the AMSE criterion
in equation~\eqref{eq:amse}: since \(\E\bm T_{n,h}\to\bm0\), their
cross term is \(o\{h^4+(nh^d)^{-1}\}\). Differentiating the criterion gives
\begin{equation*}
  4h^3\|\bm B(z)\|_2^2
  -\frac{d\,\tr\{\bm V(z)\}}{nh^{d+1}}=0,
\end{equation*}
which proves equation~\eqref{eq:optimal-bandwidth} and the stated rates.

\section{Proof of Pointwise Asymptotic Normality}
\label{app:normality-proof}

\begin{proof}[Proof of Theorem~\ref{thm:asymptotic-normality}]
For any fixed \(a\in\mathbb R^p\), consider the triangular array
\begin{equation*}
  \zeta_{i,n}
  =
  \frac{1}{\sqrt{nh^d}}a^\top\bm\psi_{i,h}(z).
\end{equation*}
Equation~\eqref{eq:appendix-score-variance} gives
\begin{equation*}
  \sum_{i=1}^n\Var(\zeta_{i,n})
  \longrightarrow
  a^\top\{\nu_0f(z)\bm\Omega(z)\}a.
\end{equation*}
The bounded kernel, compact support, and \((4+\delta)\)-moment condition imply
the Lindeberg condition. Hence the Lindeberg--Feller central limit theorem and
the Cram\'er--Wold device give
\begin{equation*}
  \frac{1}{\sqrt{nh^d}}
  \sum_{i=1}^n\bm\psi_{i,h}(z)
  \xrightarrow{d}
  N_p\{\bm0,\nu_0f(z)\bm\Omega(z)\}.
\end{equation*}
Multiplying equation~\eqref{eq:appendix-score-identity} by
\(\sqrt{nh^d}\), using equation~\eqref{eq:appendix-Gram-limit}, and applying
Slutsky's theorem proves equation~\eqref{eq:clt-population-target} with
covariance \(\bm V(z)\).

By equation~\eqref{eq:population-bias-expansion},
\begin{equation*}
  \sqrt{nh^d}
  \{\bbeta_h(z)-\bbeta(z)-h^2\bm B(z)\}
  =O\{\sqrt{nh^d}\,h^4\}.
\end{equation*}
The right-hand side converges to zero under \(nh^{d+8}\to0\), proving
equation~\eqref{eq:clt-bias-centred}.
If \(nh^{d+4}\to0\), then \(\sqrt{nh^d}h^2\bm B(z)\to0\) as well.
Slutsky's theorem therefore gives Equation~\eqref{eq:clt-undersmoothed}.
\end{proof}

\section{Curvature and Adaptive-Neighbour Proofs}
\label{app:curvature-proofs}

To obtain the small-ball volume in
Proposition~\ref{prop:curvature-contribution}, apply
Lemma~\ref{lem:radial-kernel-moment} with
\(\bm F\equiv1\) and \(K(r)=I(r<1)\). For this kernel,
\begin{equation*}
  \mu_0=v_d,
  \qquad
  \mu_2=\frac{v_d}{d+2}.
\end{equation*}
Equation~\eqref{eq:curved-kernel-moment} becomes
\begin{equation*}
  \operatorname{Vol}_\kappa\{B(z,h)\}
  =v_dh^d
  \left\{
    1-\frac{\operatorname{Scal}_\kappa}{6(d+2)}h^2
    +O(h^4)
  \right\},
\end{equation*}
which proves equation~\eqref{eq:small-ball-volume}.

For the variance correction, impose the additional conditions stated in
Proposition~\ref{prop:curvature-contribution}: \(f\), \(\bm M\),
\(\bm\Omega\), and \(\bbeta\) are constant on a fixed neighbourhood of
\(z\), and \(\|\bx_i\|\) and
\(\sigma_i^2=\E(\varepsilon_i^2\mid\bx_i,Z_i)\) are bounded there.
Set \(N_h=nh^d\), \(\bm A_{n,h}=h^{-d}\widehat{\bm\Gamma}_{n,h}\), and
\begin{equation*}
  \bm C_{n,h}=\frac{1}{nh^d}\sum_iw_i^2\sigma_i^2\bx_i\bx_i^\top.
\end{equation*}
The precise design-moment condition used here is, for some \(\eta>0\),
\begin{equation}
  \sup_n\E\|\bm A_{n,h}^{-1}\bm C_{n,h}
                 \bm A_{n,h}^{-1}\|^{1+\eta}<\infty,
  \label{eq:curvature-design-moment}
\end{equation}
where the sandwich is set to zero on singular designs for this condition.
This additional tail control is not implied by nonsingularity with
probability tending to one; it controls the inverse Gram factors in the
combination that actually determines conditional variance.
Lemma~\ref{lem:radial-kernel-moment} gives
\begin{align*}
  \bm\Gamma_h
  & =h^df\bm M
  \left\{
    \mu_0-\frac{h^2\mu_2}{6}
    \operatorname{Scal}_\kappa+O(h^4)
  \right\},\\
  \E(\bm\psi_{i,h}\bm\psi_{i,h}^\top)
  & =h^df\bm\Omega
  \left\{
    \nu_0-\frac{h^2\nu_2}{6}
    \operatorname{Scal}_\kappa+O(h^4)
  \right\}.
\end{align*}
Expanding the two population inverse Gram factors around \(\mu_0f\bm M\)
produces
\begin{equation*}
  1+
  \frac{\operatorname{Scal}_\kappa}{6}
  \left(2\frac{\mu_2}{\mu_0}-\frac{\nu_2}{\nu_0}\right)h^2
  +O(h^4),
\end{equation*}
for the corresponding population sandwich.

Here is the random-inverse justification. On nonsingular designs the
conditional mean is the constant \(\bbeta(z)\), and the conditional variance
is \(N_h^{-1}\bm A_{n,h}^{-1}\bm C_{n,h}\bm A_{n,h}^{-1}\).
Bounded local summands, independence, and positive definiteness of
\(\E\bm A_{n,h}\) imply an exponentially small probability
\(O(e^{-cN_h})\), for some \(c>0\), that
\(\|\bm A_{n,h}-\E\bm A_{n,h}\|\) exceeds half the smallest eigenvalue
of \(\E\bm A_{n,h}\). The second and fourth moments of the centred
\(\bm A_{n,h}\) and \(\bm C_{n,h}\) have orders \(N_h^{-1}\) and
\(N_h^{-2}\), respectively. On the complementary regular event, the
second-order Taylor remainder of \((\bm A,\bm C)\mapsto
\bm A^{-1}\bm C\bm A^{-1}\) therefore has expectation \(O(N_h^{-1})\).
The linear terms have zero expectation before restriction to that event.
Equation~\eqref{eq:curvature-design-moment} and H\"older's inequality
control the discarded covariance-sandwich contribution
exponentially. Consequently,
\begin{equation*}
\begin{aligned}
 \E(\bm A_{n,h}^{-1}\bm C_{n,h}\bm A_{n,h}^{-1})
  ={}&(\E\bm A_{n,h})^{-1}(\E\bm C_{n,h})
       (\E\bm A_{n,h})^{-1}\\
     &+O(N_h^{-1}).
\end{aligned}
\end{equation*}
The bounded completion on singular designs contributes only an exponentially
small moment term. Thus the relative random-inverse remainder is
\(O(N_h^{-1})=o(h^2)\) when \(nh^{d+2}\to\infty\), proving
equation~\eqref{eq:curvature-variance-correction}. Constancy of
\(\bbeta\), as well as the design moments, excludes an additional
second-order contribution from local coefficient variation. For the
bisquare kernel with \(d=2\),
\begin{equation*}
  2\frac{\mu_2}{\mu_0}-\frac{\nu_2}{\nu_0}
  =\frac{1}{6},
  \qquad
  \operatorname{Scal}_{-1}=-2,
\end{equation*}
so the relative correction is \(1-h^2/18+O(h^4)\).
This \(O(h^4)\) remainder describes the population sandwich expansion. After
the random-inverse contribution is included, the empirical variance expansion
has the \(1-h^2/18+o(h^2)\) form stated in
equation~\eqref{eq:curvature-variance-correction}.

\begin{proof}[Proof of Corollary~\ref{cor:adaptive-knn}]
Applying Lemma~\ref{lem:radial-kernel-moment} to \(\bm F=f\) and the
indicator kernel gives the local distance distribution
\begin{align*}
  P\{d_\kappa(z,Z_i)\le r\}
  ={}&v_df(z)r^d
  \left[
    1+
    \frac{r^2}{d+2}
    \left\{
      \frac{\Delta_\kappa f(z)}{2f(z)}
      -\frac{\operatorname{Scal}_\kappa}{6}
    \right\}
    +O(r^4)
  \right].
\end{align*}
Set the leading term equal to \(k/n\), which defines
\(h_{n,k}^{(0)}=\{k/(nfv_d)\}^{1/d}\), and write
\(r=h_{n,k}^{(0)}\{1+c(h_{n,k}^{(0)})^2\}\). Expanding the preceding display and equating the
second-order term to zero yields
\begin{equation*}
  c
  =
  \frac{\operatorname{Scal}_\kappa}{6d(d+2)}
  -\frac{\Delta_\kappa f(z)}{2d(d+2)f(z)}.
\end{equation*}
The binomial fluctuation of the number of observations in the ball contributes
a relative \(O_p(k^{-1/2})\) order-statistic term. This proves
equation~\eqref{eq:knn-radius-expansion}.

To justify the estimator statement, let
\(\bar h_{n,k}=F_z^{-1}(k/n)\). Continuity and local positivity of the
radial density imply
\(H_{n,k}/\bar h_{n,k}=1+O_p(k^{-1/2})\). Conditional on \(H_{n,k}=r\)
and the indices inside the ball, its \(k-1\) interior observations are
independent draws from the joint law truncated to \(d_\kappa(z,Z)<r\).
This is the regular conditional order-statistic law, valid since ties
have probability zero. The boundary observation has zero weight because the
Lipschitz compactly supported kernel satisfies \(K(1)=0\); outside
observations also have zero weight.

Write \(\ell=d_\kappa(z,Z)\). The conditional population weighted
coefficient is exactly \(\bbeta_r(z)\): division by \(F_z(r)\) cancels
between its Gram and score moments. Uniformly for \(r/\bar h_{n,k}\to1\),
the per-interior-observation moments satisfy
\begin{equation}
\begin{aligned}
 \E\{K(\ell/r)\bx\bx^\top\mid\ell<r\}&\longrightarrow
                       \frac{\mu_0}{v_d}\bm M(z),\\
 \Var\{K(\ell/r)\bx(Y-\bx^\top\bbeta_r)\mid\ell<r\}
     &\longrightarrow\frac{\nu_0}{v_d}\bm\Omega(z).
\end{aligned}
\label{eq:adaptive-truncated-moments}
\end{equation}
Equation~\eqref{eq:adaptive-truncated-moments} follows by dividing the
normal-coordinate kernel expansions by \(F_z(r)\sim f(z)v_dr^d\).
Uniform local moment bounds give the conditional Lindeberg condition and
Gram convergence for \(k-1\) draws. Thus the conditional centred estimation
error is \(O_p(k^{-1/2})\), with normal covariance constant
\begin{equation}
 \bm V_k(z)=\frac{\nu_0v_d}{\mu_0^2}
             \bm M(z)^{-1}\bm\Omega(z)\bm M(z)^{-1}.
 \label{eq:adaptive-variance-constant}
\end{equation}
Equation~\eqref{eq:adaptive-variance-constant} is the covariance at scale
\(\sqrt{k}\); the density factor cancels because the neighbourhood fixes
probability mass. The conditional limits are uniform on any shrinking
relative-radius interval containing \(H_{n,k}\) with probability tending
to one, so integration over the radius gives the same unconditional limit.
Differentiating the smooth population moment integrals gives
\(\partial_r\bbeta_r=2r\bm B+O(r^3)\). Consequently
\(\bbeta_{H_{n,k}}-\bbeta_{\bar h_{n,k}}
=O_p\{(h_{n,k}^{(0)})^2k^{-1/2}\}\); this shift is negligible at
\(\sqrt{k}\) scale. The population bias is
\(\bbeta_{\bar h_{n,k}}-\bbeta=(h_{n,k}^{(0)})^2\bm B
+O\{(h_{n,k}^{(0)})^4\}\).

The additional uniform-integrability condition in the corollary converts
this limit into unconditional covariance \(\bm V_k/k+o(k^{-1})\), including
the singular-design completion. It also controls the AMSE cross term, as in
Appendix~\ref{app:consistency-proofs}. Balancing squared smoothing bias
\((k/n)^{4/d}\) and variance \(k^{-1}\) gives the stated optimal order when
\(\bm B\ne0\). No unrestricted claim that the sample's expectation bias
equals its population smoothing bias is needed.
\end{proof}

\section{Proof of the Flat-Curvature Limit}
\label{app:flat-limit-proof}

\begin{proof}[Proof of Theorem~\ref{thm:flat-limit}]
Let \(a_m=\sqrt{-\kappa_m}\). In normal polar coordinates on
\(\mathbb H_{\kappa_m}^d\), the metric and volume density are
\begin{equation*}
  ds_m^2
  =dr^2+\left\{\frac{\sinh(a_mr)}{a_m}\right\}^2d\omega^2,
  \qquad
  J_m(r)
  =\left\{\frac{\sinh(a_mr)}{a_mr}\right\}^{d-1}.
\end{equation*}
On every bounded normal-coordinate set,
\begin{equation*}
  \frac{\sinh(a_mr)}{a_m}
  =r+\frac{a_m^2r^3}{6}+O(a_m^4r^5),
  \qquad
  J_m(r)=1+O(|\kappa_m|r^2),
\end{equation*}
uniformly in \(r\). The explicit smooth metric coefficients, their inverses,
and the derivatives needed in Equation~\eqref{eq:laplace-beltrami-definition}
converge to their Euclidean counterparts on compact normal-coordinate sets
(with continuous values at the origin). The corresponding distances and
Laplace--Beltrami operators on \(C^2\) fields therefore converge.
Because \(K\) is bounded and
compactly supported, dominated convergence applies to every kernel moment in
Lemma~\ref{lem:radial-kernel-moment}. The scalar-curvature terms vanish because
\(\operatorname{Scal}_{\kappa_m}=d(d-1)\kappa_m\to0\).

The assumed \(C^2\) convergence of the pulled-back fields, together with
uniform bounds on fourth derivatives and moments, implies convergence of
\(\bm B_m(z_m)\), \(\bm V_m(z_m)\), and the small-ball probability to their
Euclidean counterparts. The condition
\(|\kappa_m|h_m^2\to0\) removes the curvature correction uniformly over the
smoothing neighbourhood. Together with \(n_mh_m^{d+8}\to0\), Slutsky's theorem
applied to the triangular-array argument in
Appendix~\ref{app:normality-proof} gives convergence of the bias-centred
limiting law. The same small-ball expansion, under the stated conditions on
\(k_m\), gives the Euclidean limit of the adaptive-radius expansion.
\end{proof}

\section{Proofs for Generated Lorentz Representations}
\label{app:generated-representation-proofs}

\begin{proof}[Proof of Proposition~\ref{prop:lorentz-isometry-invariance}]
Write \(\bm J=\operatorname{diag}(-1,\bm I_d)\) for the ambient Lorentz
metric matrix. For every
\(L\in O^+(1,d)\),
\begin{equation*}
  L^\top\bm J L=\bm J.
\end{equation*}
Consequently,
\begin{equation*}
  \langle Lz,Lq\rangle_L=\langle z,q\rangle_L
  \quad\text{and}\quad
  d_\kappa(Lz,Lq)=d_\kappa(z,q).
\end{equation*}
Every kernel weight is therefore unchanged under the simultaneous
transformation. The weighted Gram matrices and score vectors are identical,
so the local coefficient estimates, fitted values, and predictions are
identical as well.
\end{proof}

\begin{proof}[Proof of Proposition~\ref{prop:first-variation-distance}]
Hyperbolic space is complete and simply connected with negative sectional
curvature, so distinct points are joined by a unique minimizing geodesic. The
first variation formula for its length gives
\begin{equation*}
  \left.\frac{d}{dt}d_\kappa(z_t,q_t)\right|_{t=0}
  =-g_z\{a,\dot\gamma_{zq}(0)\}
   -g_q\{b,\dot\gamma_{qz}(0)\},
\end{equation*}
where \(\gamma_{zq}\) and \(\gamma_{qz}\) are the unit-speed minimizing
geodesics from \(z\) to \(q\) and from \(q\) to \(z\), respectively. Their
initial velocities are \(u_{zq}\) and \(u_{qz}\), which proves
Equation~\eqref{eq:first-variation-distance}. A common Lorentz isometry
preserves the distance for every \(t\), so its derivative is zero.
\end{proof}

\begin{proof}[Proof of Theorem~\ref{thm:generated-distance-expansion}]
Assumption~\ref{ass:generated-representation} is in force. Write
\(\widehat{\bm\Gamma}_{n,h}^{a}=n^{-1}\sum_iw_i^a\bx_i\bx_i^\top\)
for \(a\in\{O,E\}\), and let \(r_i^O(z)\) be the oracle residual defined
in that assumption.
Let
\begin{equation*}
  w_i^O=K(d_i/h),
  \qquad
  w_i^E=K(\widehat d_i/h),
  \qquad
  \Delta_i=\widehat d_i-d_i.
\end{equation*}
Because \(K'\) is Lipschitz, Taylor's theorem gives, uniformly over the
contributing observations,
\begin{equation}
  w_i^E-w_i^O
  =K'(d_i/h)\frac{\Delta_i}{h}+q_{i,n},
  \qquad
  |q_{i,n}|\leq C\rho_n(z)^2.
  \label{eq:appendix-weight-linearization}
\end{equation}
Here \(K\) is defined on \([0,\infty)\) and extended by zero beyond 1,
with Lipschitz derivative as required in the assumption; this includes the
bisquare kernel. The union of the oracle and perturbed neighbourhoods is contained in a
geodesic ball of radius \(h\{1+o_p(1)\}\), whose empirical mass is
\(O_p(h^d)\).

The oracle normal equations imply
\begin{equation*}
  \frac1n\sum_{i=1}^n w_i^O\bx_i r_i^O(z)=\bm0.
\end{equation*}
Subtracting the oracle equations from the perturbed equations therefore gives
the exact identity
\begin{equation}
  \widehat\bbeta_h^{\,E}(z)-\widehat\bbeta_h^{\,O}(z)
  =\{\widehat{\bm\Gamma}_{n,h}^{E}(z)\}^{-1}
   \frac1n\sum_{i=1}^n(w_i^E-w_i^O)\bx_i r_i^O(z).
  \label{eq:appendix-exact-generated-identity}
\end{equation}
The local moment condition and Equation
\eqref{eq:appendix-weight-linearization} yield
\begin{align*}
  h^{-d}\|\widehat{\bm\Gamma}_{n,h}^{E}
       -\widehat{\bm\Gamma}_{n,h}^{O}\|&=O_p\{\rho_n(z)\},\\
  \left\|\frac1n\sum_{i=1}^n q_{i,n}\bx_i r_i^O(z)\right\|
       &=O_p\{h^d\rho_n(z)^2\}.
\end{align*}
Uniform nonsingularity and
\begin{equation*}
  \bm A^{-1}-\bm B^{-1}
  =\bm A^{-1}(\bm B-\bm A)\bm B^{-1}
\end{equation*}
then imply
\begin{equation*}
  \|\{\widehat{\bm\Gamma}_{n,h}^{E}\}^{-1}
    -\{\widehat{\bm\Gamma}_{n,h}^{O}\}^{-1}\|
  =O_p\{h^{-d}\rho_n(z)\}.
\end{equation*}
Substituting these bounds into Equation
\eqref{eq:appendix-exact-generated-identity} gives
\begin{equation*}
  \widehat\bbeta_h^{\,E}(z)-\widehat\bbeta_h^{\,O}(z)
  =\{\widehat{\bm\Gamma}_{n,h}^{O}(z)\}^{-1}
   \frac1n\sum_{i=1}^n
   K'(d_i/h)\frac{\Delta_i}{h}\bx_i r_i^O(z)
   +O_p\{\rho_n(z)^2\},
\end{equation*}
which is Equation~\eqref{eq:generated-distance-expansion}. The displayed
linear term is \(O_p\{\rho_n(z)\}\), proving the final assertion.
\end{proof}

\begin{proof}[Proof of Theorem~\ref{thm:joint-representation-limit}]
Theorem~\ref{thm:generated-distance-expansion} and the condition
\(\sqrt{nh^d}\rho_n(z)^2\xrightarrow{p}0\) give
\begin{align*}
  &\sqrt{nh^d}\{\widehat\bbeta_h^{\,E}(z)
    -\bbeta(z)-h^2\bm B(z)\}\\
  &\quad=
  \sqrt{nh^d}\{\widehat\bbeta_h^{\,O}(z)-\bbeta_h(z)\}
  +\sqrt{nh^d}\bm R_{n,h}(z)
  +\sqrt{nh^d}\{\bbeta_h(z)-\bbeta(z)-h^2\bm B(z)\}
  +o_p(1).
\end{align*}
The last term before the remainder is \(o(1)\) by
Equation~\eqref{eq:population-bias-expansion} and \(nh^{d+8}\to0\).
Equation~\eqref{eq:joint-score-limit}, the continuous mapping theorem, and
Slutsky's theorem now prove Equation~\eqref{eq:joint-representation-limit}.
If the scaled representation term is \(o_p(1)\), it disappears from the
limiting law. No Gaussian assumption is needed for the addition of the two
joint limits. The bound in Theorem~\ref{thm:generated-distance-expansion}
shows that \(\sqrt{nh^d}\rho_n(z)\to0\) is a sufficient, but not necessary,
condition for this first-stage term to vanish. When the joint limit is
Gaussian, adding its two components adds their covariance blocks and both
cross blocks, giving Equation~\eqref{eq:joint-gaussian-covariance}.
\end{proof}

\begin{proof}[Proof of Corollary~\ref{cor:generated-target-location}]
In normal coordinates at \(z\), covariant Taylor expansion gives
\begin{equation}
  \bbeta\{\exp_z(v_n)\}
  =\bbeta(z)+\bm D\bbeta(z)[v_n]_z+O_p(\|v_n\|^2).
  \label{eq:appendix-target-taylor}
\end{equation}
Smoothness also gives
\(\bm B\{\exp_z(v_n)\}=\bm B(z)+O_p(\|v_n\|)\).
The corollary's equicontinuity condition transfers
the centred estimation error from \(z\) to \(\widehat z=\exp_z(v_n)\).
The joint convergence assumption gives
\(\sqrt{nh^d}[v_n]_z=O_p(1)\), so
\(\sqrt{nh^d}h^2\|v_n\|=o_p(1)\). The quadratic target remainder is
negligible by \(\sqrt{nh^d}\|v_n\|^2\to0\), and the uniform
\(O(h^4)\) population-bias remainder is negligible by \(nh^{d+8}\to0\).
Combining these observations with Equation~\eqref{eq:appendix-target-taylor}
yields
\begin{align*}
  &\sqrt{nh^d}\{\widehat\bbeta_h(\widehat z)
  -\bbeta(z)-h^2\bm B(z)\}\\
  &\qquad=
  \sqrt{nh^d}\{\widehat\bbeta_h(z)-\bbeta_h(z)\}
  +\bm D\bbeta(z)\sqrt{nh^d}[v_n]_z+o_p(1).
\end{align*}
The joint convergence assumption and Slutsky's theorem prove
Equation~\eqref{eq:generated-target-limit}. Under independence, covariance
addition gives Equation~\eqref{eq:generated-target-variance}.
Without independence, writing \(\bm C_{YU}=\Cov(\bm G_Y,U)\), a
\(p\times d\) matrix, the covariance is
\(\bm V+\bm D\bbeta\,\bm\Sigma_Z\,\bm D\bbeta^\top
+\bm C_{YU}\bm D\bbeta^\top+\bm D\bbeta\bm C_{YU}^\top\).
This covariance statement requires finite second moments; normality of the
sum additionally requires a jointly Gaussian limit.
\end{proof}

\begin{proof}[Proof of Corollary~\ref{cor:representation-stability-knn}]
On the event that both distance vectors use the same strictly positive
index set, the \(k\)th order statistic is one-Lipschitz with respect to the
sup norm. The corollary assumes this event has probability tending to one;
dropping an unequal set of zero distances would not justify the bound. Hence
\begin{equation*}
  |H_{n,k}^{E}(z)-H_{n,k}^{O}(z)|\leq\delta_n.
\end{equation*}
Corollary~\ref{cor:adaptive-knn} gives
\(H_{n,k}^{O}(z)/h_{n,k}^{(0)}(z)\xrightarrow{p}1\). Because
\(\delta_n=o_p\{h_{n,k}^{(0)}(z)\}\), the order-statistic bound implies
\(H_{n,k}^{E}(z)/h_{n,k}^{(0)}(z)\xrightarrow{p}1\) as well. For every observation
that receives positive weight under either geometry,
\begin{align*}
 \left|
 \frac{\widehat d_i}{H_{n,k}^{E}(z)}
 -\frac{d_i}{H_{n,k}^{O}(z)}
 \right|
 &\leq
 \frac{|\widehat d_i-d_i|}{H_{n,k}^{E}(z)}
 +
 \frac{d_i|H_{n,k}^{E}(z)-H_{n,k}^{O}(z)|}
      {H_{n,k}^{E}(z)H_{n,k}^{O}(z)}  \\
 &=O_p\left\{\frac{\delta_n}{h_{n,k}^{(0)}(z)}\right\}.
\end{align*}
The union of the two contributing neighbourhoods contains \(O_p(k)\)
observations. Lipschitz continuity of \(K\), the local moment conditions, and
uniform nonsingularity therefore give the adaptive analogue of the Gram and
score perturbation bounds used in the proof of
Theorem~\ref{thm:generated-distance-expansion}. The inverse identity then yields
equation~\eqref{eq:representation-stability-knn-rate}. The adaptive stochastic
scale is \(k^{-1/2}\), so the stated consistency follows directly. Whenever
the oracle adaptive estimator admits a \(\sqrt{k}\)-scaled limiting law, the
first-order equivalence condition and Slutsky's theorem transfer that law to
the estimated-representation estimator.
\end{proof}

\section{Secondary Conditional Diagnostics}
\label{app:secondary-diagnostics}

For fixed coordinates, kernel, and bandwidths, write
\(\widehat\bbeta(Z_i)=\bA_i\by\),
\(\bs_i^\top=\bx_i^\top\bA_i\), and
\(\widehat\by=\bS\by\). Under the homoskedastic uncorrelated error model,
let \(\bm R_H=(\bI_n-\bS)^\top(\bI_n-\bS)\). The residual
degrees of freedom and expected residual sum of squares are
\begin{equation}
\begin{aligned}
 \delta_1&=\tr(\bm R_H)=n-2\tr(\bS)+\tr(\bS^\top\bS),\\
 \E(RSS_H\mid\bX,\bZ)&=
       \|(\bI_n-\bS)\bm m\|^2+\sigma^2\delta_1.
\end{aligned}
\label{eq:residual-moment}
\end{equation}
Equation~\eqref{eq:residual-moment} follows by expanding the residual
quadratic form. It separates smoothing bias from residual noise and shows
why division by residual degrees of freedom is not generally unbiased for
\(\sigma^2\). For \(\delta_1>0\), the working plug-in quantities are
\begin{equation}
  \widehat\sigma^2\bA_i\bA_i^\top,
  \qquad
  \widehat\sigma^2=
  \frac{\|\by-\widehat\by\|_2^2}
       {\delta_1}.
 \label{eq:working-covariance-estimate}
\end{equation}
Equation~\eqref{eq:working-covariance-estimate} uses the special case
\(\bm\Sigma=\sigma^2\bI_n\) of
Equation~\eqref{eq:conditional-moments}. Response-selected bandwidths,
smoothing bias, and residual dependence are not removed by treating the
realized neighbourhoods as fixed. Pointwise tests and Benjamini--Hochberg
classifications are therefore exploratory; the adjustment alone does not
establish calibrated post-selection p-values or false-discovery control.

Moment-matched F1--F3 diagnostics follow Leung et al.~\citep{leung2000}, with the local
smoother built from Lorentz-hyperbolic distance. Let \(RSS_0\) and \(RSS_H\)
be the OLS and LHWR residual sums of squares, respectively, and
\(\nu_{\rm OLS}=n-p>0\) the OLS residual degrees of freedom. Write
\begin{equation*}
  \bm R_H=(\bI_n-\bS)^\top(\bI_n-\bS),\qquad
  \delta_1=\tr(\bm R_H),\qquad \delta_2=\tr(\bm R_H^2).
\end{equation*}
Then
\begin{equation*}
  F_1=\frac{RSS_H/\delta_1}{RSS_0/\nu_{\rm OLS}}
  \ \dot\sim\ F_{\delta_1^2/\delta_2,\nu_{\rm OLS}}.
\end{equation*}
With \(\eta_1=\nu_{\rm OLS}-\delta_1\) and
\(\eta_2=\nu_{\rm OLS}-2\delta_1+\delta_2\),
\begin{equation*}
  F_2=\frac{(RSS_0-RSS_H)/\eta_1}{RSS_0/\nu_{\rm OLS}}
  \ \dot\sim\ F_{\eta_1^2/\eta_2,\nu_{\rm OLS}}.
\end{equation*}
For coefficient index \(k\in\{1,\ldots,p\}\) (not the neighbour-count
parameter in this subsection), let \(\bm B_k\by\) collect the \(n\) local
estimates. Define the \(n\times n\) centring matrix
\(\bm C_0=\bI_n-n^{-1}\mathbf1_n\mathbf1_n^\top\), and
\(\bm C_k=n^{-1}\bm B_k^\top\bm C_0\bm B_k\). Here \(\mathbf1_n\)
is the vector of ones. Defining
\(\gamma_{1k}=\tr(\bm C_k)\),
\(\gamma_{2k}=\tr(\bm C_k^2)\), and
\(V_k^2=n^{-1}\by^\top\bm B_k^\top\bm C_0\bm B_k\by\), the coefficient
stationarity diagnostic is
\begin{equation*}
  F_{3,k}=\frac{V_k^2/\gamma_{1k}}{RSS_H/\delta_1}
  \ \dot\sim\ F_{\gamma_{1k}^2/\gamma_{2k},\delta_1^2/\delta_2}.
\end{equation*}
The symbol \(\dot\sim\) denotes a moment-matched reference approximation,
not an exact F law: the numerator and denominator quadratic forms are
generally dependent even under Gaussian errors. F1 uses the lower tail;
F2 and F3 use upper tails. All required denominator sums of squares and
moment-based degrees of freedom must be positive; otherwise the corresponding
reference p-value is undefined. The F2 approximation concerns an improvement
in fit and is not evidence for local structure when the numerator is negative.
The stationary independent Gaussian model with fixed weights motivates these
reference laws. In particular, calibration with all coefficients stationary
does not prove validity for testing one stationary coefficient when nuisance
coefficient fields vary. These limitations accompany the numerical
calibration below.

For local design assessment, let \(q\) be the number of non-intercept
predictors and \(\widetilde{\bX}_i\) contain them after centring by their
weighted local means and scaling to unit weighted sums of squares. This
scaling is only for the collinearity diagnostic, not a refit or transformation
of the regression model. The local condition number and variance
inflation factor are
\begin{equation*}
  \mathrm{CN}_i=\sqrt{\lambda_{i,1}/\lambda_{i,q}},
  \qquad
  \mathrm{VIF}_{ik}=\{1-R_{ik}^2\}^{-1},
\end{equation*}
where \(\lambda_{i,1}\geq\cdots\geq\lambda_{i,q}>0\) are the eigenvalues of
\(\widetilde{\bX}_i^\top\bW_i\widetilde{\bX}_i\), and \(R_{ik}^2\) is from
the weighted regression of predictor \(k\) on the remaining predictors
\citep{wheeler2005}.

For row-standardised diagnostic weights \(w_{ij}^{(H)}\) with zero diagonal,
the residual Lorentz-hyperbolic Moran statistic is
\begin{equation*}
  I_H=\frac{n}{S_0}
  \frac{\sum_{i,j}w_{ij}^{(H)}(e_i-\bar e)(e_j-\bar e)}
       {\sum_i(e_i-\bar e)^2},
  \qquad \E_{\rm perm}(I_H\mid\be,\bm W^{(H)})=-\frac{1}{n-1},
\end{equation*}
where \(e_i=Y_i-\widehat Y_i\), \(\bar e=n^{-1}\sum_i e_i\), and
\(S_0=\sum_{i,j}w_{ij}^{(H)}>0\) \citep{moran1950}. The expectation is
over uniform permutations of a fixed nonconstant residual vector, with
\(n>1\) and the diagnostic weights fixed. It is not the sampling expectation
of residuals estimated by LHWR. The local statistic is
\begin{equation*}
  I_i^{(H)}=\frac{(e_i-\bar e)\sum_jw_{ij}^{(H)}(e_j-\bar e)}
  {n^{-1}\sum_i(e_i-\bar e)^2}
\end{equation*}
and gives the usual high--high, low--low, high--low, and low--high labels
\citep{anselin1995}. Because smoothing makes fitted residuals nonexchangeable,
ordinary residual permutations are descriptive rather than exact post-fit
tests.

\section{Finite-Sample Diagnostic Calibration}
\label{app:diagnostic-calibration}

The calibration experiment uses 200 replications with \(n=300\),
\(R_{\max}=6\), fixed \(k=60\), nominal level 0.05, and 199 Moran
permutations. It compares stationary coefficients with independent errors,
localised coefficients with independent errors, and stationary coefficients
with Lorentz-neighbourhood correlated errors.

\IfFileExists{simulation_results/inference_mc_v2/tables/inference_compact_summary.tex}{%
  \input{simulation_results/inference_mc_v2/tables/inference_compact_summary.tex}
}{%
  \noindent The generated calibration table was not found.
}

Under stationary independent errors, F1 and F2 are conservative (rejection
rates 0.000 and 0.005), while F3 rates range from 0.040 to 0.055. Under
localised heterogeneity, the global tests reject in every replication and F3
rates range from 0.970 to 1.000. Correlated errors inflate several rejection
rates and reduce pointwise coverage, which motivates the qualifications in the
main-text discussion.

\FloatBarrier
\section{Additional WDI Diagnostics}
\label{app:additional-wdi-diagnostics}

For the WDI fit, F1 and F2 give p-values 0.026 and below 0.001. The F3
p-values are 0.123 for the intercept, 0.006 for unemployment, and 0.055 for
internet use. Median and maximum local condition numbers are 1.31 and 2.02,
and all local variance inflation factors are below 1.58. With \(k=40\), OLS
residuals have \(I_H=0.0717\) (permutation p-value 0.020), whereas LHWR
residuals have \(I_H=-0.0383\) (p-value 0.140); these fitted-residual summaries
are descriptive for the reason given above.

\end{document}

%% file: tables/tas_simulation_core.tex
\begin{table}[!htbp]
\centering
\caption{Monte Carlo comparison for the matched designs. Entries are means
(Monte Carlo standard deviations) over 100 replications after averaging over
\(R_{\max}\in\{4,6,8,10\}\). Oracle errors minimise
\(\mathrm{RMSE}_{\beta}\) over the candidate bandwidths and are unavailable
in practice. CV denotes bandwidth selection; the reported errors measure
recovery of the true signal and coefficients at the simulated design points.}
\label{tab:simulation-core-results}
\begingroup
\setlength{\tabcolsep}{5pt}
\begin{tabular}{lrrrr}
\toprule
Method & CV \(k\) & CV \(\mathrm{RMSE}_{\eta}\) &
CV \(\mathrm{RMSE}_{\beta}\) & Oracle \(\mathrm{RMSE}_{\beta}\) \\
\midrule
\multicolumn{5}{l}{\textit{Localized surface, \(k_0=20\)}} \\
OLS & -- & 1.325 (0.090) & 1.365 (0.043) & -- \\
LHWR & 95.2 (26.8) & 0.808 (0.051) & 1.034 (0.040) & 1.014 \\
Poincar\'e-E & 56.9 (27.0) & 0.862 (0.060) & 1.070 (0.040) & 1.047 \\
Tangent-E & 74.6 (26.6) & 0.936 (0.066) & 1.156 (0.037) & 1.130 \\
\addlinespace
\multicolumn{5}{l}{\textit{Smooth surface, \(k_0=80\)}} \\
OLS & -- & 0.711 (0.032) & 0.713 (0.020) & -- \\
LHWR & 104.3 (19.4) & 0.379 (0.018) & 0.453 (0.017) & 0.448 \\
Poincar\'e-E & 40.4 (6.4) & 0.360 (0.013) & 0.427 (0.011) & 0.420 \\
Tangent-E & 42.1 (5.2) & 0.333 (0.014) & 0.395 (0.014) & 0.390 \\
\bottomrule
\end{tabular}
\endgroup
\end{table}

%% file: tables/tas_generated_representation.tex
\begin{table}[!htbp]
\centering
\caption{Generated-target uncertainty at \(n=16000\) over 1000
replications. Conditional intervals treat the representation as fixed;
adjusted intervals add the target-location covariance. Both use the known
second-order smoothing bias and analytic leading variances. The nominal
coverage is 0.95; the dominant regime is a stress test beyond the finite
joint-limit setting.}
\label{tab:generated-representation-validation}
\begingroup
\setlength{\tabcolsep}{4pt}
\begin{tabular}{llrrrrr}
\toprule
Regime & Coef. & Emp. SD & Cond. SE & Adj. SE & Cond. cov. & Adj. cov. \\
\midrule
Fixed & \(\beta_1\) & 0.081 & 0.074 & 0.074 & 0.920 & 0.920 \\
Fixed & \(\beta_2\) & 0.081 & 0.074 & 0.074 & 0.926 & 0.926 \\
Negligible first stage & \(\beta_1\) & 0.083 & 0.074 & 0.075 & 0.911 & 0.925 \\
Negligible first stage & \(\beta_2\) & 0.082 & 0.074 & 0.075 & 0.924 & 0.926 \\
First-order first stage & \(\beta_1\) & 0.095 & 0.074 & 0.087 & 0.885 & 0.935 \\
First-order first stage & \(\beta_2\) & 0.090 & 0.074 & 0.082 & 0.891 & 0.928 \\
Representation dominant & \(\beta_1\) & 0.198 & 0.074 & 0.192 & 0.526 & 0.935 \\
Representation dominant & \(\beta_2\) & 0.159 & 0.074 & 0.152 & 0.635 & 0.938 \\
\bottomrule
\end{tabular}
\endgroup
\end{table}

%% file: tables/tas_wdi_loocv.tex
\begin{table}[!htbp]
\centering
\caption{\label{tab:wdi-feature-map-predictive}LOOCV conditional on the fixed WDI representation for the common 141-country sample. All local methods use the same kernel, candidate grid and training-fold AICc rule; $\overline{k}$ is the mean selected $k$. Differences are competitor loss minus LHWR loss; positive values favor LHWR, and brackets are 95\% paired-bootstrap percentile intervals for the computed prediction errors, without refitting models.}
\centering
\setlength{\tabcolsep}{3pt}
\begin{tabular}[t]{lrrrrrr}
\toprule
Method & $\overline{k}$ & $\mathrm{RMSE}$ & $\Delta\mathrm{RMSE}$ [95\%] & $\mathrm{MAE}$ & $\Delta\mathrm{MAE}$ [95\%] & $R^2_{\mathrm{pred}}$\\
\midrule
OLS & -- & 2.446 & 0.230 [0.053, 0.389] & 1.864 & 0.186 [0.038, 0.327] & 0.054\\
GWR & 58 & 2.519 & 0.302 [0.040, 0.568] & 1.879 & 0.201 [0.023, 0.380] & -0.003\\
Poincar\'e-E & 49 & 2.246 & 0.029 [0.002, 0.059] & 1.707 & 0.029 [-0.006, 0.064] & 0.203\\
Tangent-E & 48 & 2.236 & 0.020 [-0.003, 0.044] & 1.703 & 0.025 [-0.008, 0.056] & 0.209\\
LHWR & 41 & 2.216 & -- & 1.678 & -- & 0.223\\
\bottomrule
\end{tabular}
\end{table}

%% file: simulation_results/inference_mc_v2/tables/inference_compact_summary.tex
\begin{table}[!t]
\centering
\caption{Finite-sample calibration at nominal level 0.05. Panel A gives conditional F-test rejection rates; Panel B gives pointwise 95\% interval coverage and zero-coefficient rejection for $x_2$; Panel C separates LH-Moran behavior for fitted residuals from that for the generating errors.}
\label{tab:inference-calibration}
\footnotesize
\setlength{\tabcolsep}{4pt}
\textbf{Panel A: conditional F-test rejection rates}\par\smallskip
\begin{tabular}{lrrrrr}
\toprule
Scenario & F1 & F2 & F3: int. & F3: $x_1$ & F3: $x_2$ \\
\midrule
Stationary, independent & 0.000 & 0.005 & 0.040 & 0.055 & 0.055 \\
Localized, independent & 1.000 & 1.000 & 1.000 & 0.995 & 0.970 \\
Stationary, correlated & 0.155 & 0.735 & 0.970 & 0.065 & 0.065 \\
\bottomrule
\end{tabular}
\par\medskip
\textbf{Panel B: pointwise interval coverage and zero rejection}\par\smallskip
\begin{tabular}{lrrrrr}
\toprule
Scenario & Cov. int. & Cov. $x_1$ & Cov. $x_2$ & Raw zero rej. & BH zero rej. \\
\midrule
Stationary, independent & 0.945 & 0.948 & 0.952 & 0.048 & 0.001 \\
Localized, independent & 0.895 & 0.903 & 0.932 & -- & -- \\
Stationary, correlated & 0.731 & 0.950 & 0.948 & 0.052 & 0.001 \\
\bottomrule
\end{tabular}
\par\medskip
\textbf{Panel C: LH-Moran diagnostic}\par\smallskip
\begin{tabular}{lrrrr}
\toprule
Scenario & Fitted mean $I$ & Fitted rej. & Error mean $I$ & Error rej. \\
\midrule
Stationary, independent & -0.045 & 0.985 & -0.002 & 0.070 \\
Localized, independent & -0.040 & 0.855 & -0.002 & 0.035 \\
Stationary, correlated & -0.027 & 0.230 & 0.102 & 0.985 \\
\bottomrule
\end{tabular}
\end{table}